\documentclass{lmcs}


\usepackage{amsmath}
\usepackage{amssymb}
\usepackage{latexsym}
\usepackage{bussproofs}
\usepackage{cmll}

\usepackage{tikz}
\usetikzlibrary{matrix}
\usepackage{pgfplots}

\newcommand\CMLLPAR{
\usepackage{cmll}
\newcommand\IPar{\mathord{\parr}}
}

\CMLLPAR
\usepackage{stmaryrd}

\newtheorem*{lemma*}{Lemma}
\newtheorem*{proposition*}{Proposition}

\newcommand{\Endproof}{
  \ifmmode 
  \else \leavevmode\unskip\penalty9999 \hbox{}\nobreak\hfill
  \fi
  \quad\hbox{$\Box$}
  \par\medskip}

\newcommand\Eqref[1]{(\ref{#1})}

\newcommand\Eg{\textsl{e.g.}}

\renewcommand{\phi}{\varphi}
\renewcommand\epsilon{\varepsilon}

\newcommand{\Implies}{\Rightarrow}

\newcommand\Equiv{\Leftrightarrow}
\newcommand{\St}{\mid}

\newcommand{\Infi}{\wedge}

\newcommand{\Supi}{\vee}

\renewcommand{\Bot}{{\mathord{\perp}}}

\newcommand\Seqempty{\Tuple{}}

\newcommand\cC{\mathcal{C}}
\newcommand\cD{\mathcal{D}}

\newcommand\cL{\mathcal{L}}

\newcommand\Fini{{\mathrm{fin}}}

\newcommand\Part[1]{{\mathcal P}\left({#1}\right)}

\newcommand{\Linarrow}{\multimap}

\newcommand\Myleft{}
\newcommand\Myright{}

\newcommand\Web[1]{\Myleft|{#1}\Myright|}

\newcommand\Supp[1]{\operatorname{\mathsf{supp}}({#1})}

\newcommand\Mset[1]{[{#1}]}

\newcommand\ITens{\otimes}
\newcommand\Tens[2]{{#1}\ITens{#2}}

\newcommand\IWith{\mathrel{\&}}

\newcommand\Orth[2][]{#2^{\Bot_{#1}}}
\newcommand\Orthp[2][]{(#2)^{\Bot_{#1}}}

\newcommand\Bwith{\mathop{\&}}
\newcommand\Bplus{\mathop\oplus}
\newcommand\Bunion{\mathop\cup}

\newcommand\Biorth[1]{#1^{\Bot\Bot}}

\newcommand\One{1}

\newcommand\Card[1]{\#{#1}}

\newcommand\FamRestr[2]{{#1}|_{#2}}

\newcommand\Locun[1]{1^J}

\newcommand\Isom\simeq

\newcommand\NUCS{\mathbf{nCoh}}

\newcommand\Comp{\mathrel\circ}

\newcommand\Limpl[2]{{#1}\Linarrow{#2}}

\newcommand\Nat{{\mathbb{N}}}

\newcommand\Biind[2]{\genfrac{}{}{0pt}{1}{#1}{#2}}

\newcommand\Snat{\mathsf N}
\newcommand\Snatcoh{\mathsf N^\NUCS}

\newcommand\App[2]{({#1}){#2}}

\newcommand\Abst[3]{\lambda#1^{#2}\,{#3}}

\newcommand\List[3]{#1_{#2},\dots,#1_{#3}}

\newcommand\Kronecker[2]{\delta_{{#1},{#2}}}

\newcommand\Subst[3]{{#1}\left[{#2}/{#3}\right]}

\newcommand\Factor[1]{{#1}!}
\newcommand\Binom[2]{\genfrac{(}{)}{0pt}{}{#1}{#2}}

\newcommand\Real{\mathbb{R}}
\newcommand\Realp{\mathbb{R}_{\geq 0}}
\newcommand\Realpto[1]{(\Realp)^{#1}}
\newcommand\Realto[1]{\Real^{#1}}

\newcommand\Realpc{\overline{\Realp}}
\newcommand\Realpcto[1]{\Realpc^{#1}}

\newcommand\Intercc[2]{[#1,#2]}
\newcommand\Interoc[2]{(#1,#2]}
\newcommand\Interco[2]{[#1,#2)}

\newcommand\Rational{\mathbb Q}

\newcommand\Bcanon[1]{e_{#1}}

\newcommand\Mfin[1]{\mathcal M_\Fini({#1})}

\newcommand\Diracm{\delta}

\newcommand\REL{\operatorname{\mathbf{Rel}}}

\newcommand\Norm[1]{\|{#1}\|}
\newcommand\Normsp[2]{\|{#1}\|_{#2}}

\newcommand\Rel[1]{\mathrel{#1}}

\newcommand\Redst[1]{\mathop{\mathsf{Red}}}

\newcommand\Tuple[1]{\langle{#1}\rangle}

\newcommand\Msetofsubst[1]{\bar F}

\newcommand\Pcoh[1]{\mathsf P{#1}}
\newcommand\Pcohp[1]{\Pcoh{(#1)}}
\newcommand\Pcohc[1]{\overline{\mathsf P}{#1}}
\newcommand\Pcohcp[1]{\Pcohc{(#1)}}

\newcommand\Base[1]{\mathsf e_{#1}}
\newcommand\Matapp[2]{{#1}\Compl{#2}}
\newcommand\Matappa[2]{{#1}\cdot{#2}}

\newcommand\PCOH{\mathbf{Pcoh}}

\newcommand\Absval[1]{\left|{#1}\right|}

\newcommand\Retri\zeta
\newcommand\Retrp\rho

\newcommand\Impl[2]{{#1}\Rightarrow{#2}}

\newcommand\Tsem[1]{\llbracket{#1}\rrbracket}
\newcommand\Tsemrel[1]{\llbracket{#1}\rrbracket^{\REL}}
\newcommand\Tsemcoh[1]{\llbracket{#1}\rrbracket^{\NUCS}}

\newcommand\Psem[2]{\llbracket{#1}\rrbracket_{#2}}
\newcommand\Psemrel[2]{\llbracket{#1}\rrbracket^{\REL}_{#2}}
\newcommand\Psemcoh[2]{\llbracket{#1}\rrbracket^{\NUCS}_{#2}}

\newcommand\Tnat\iota
\newcommand\Fix[1]{\operatorname{\underline{\mathsf{fix}}}(#1)}
\newcommand\If[3]{\operatorname{\underline{\mathsf{if}}}(#1,#2,#3)}
\newcommand\Pred[1]{\operatorname{\underline{\mathsf{pred}}}(#1)}
\newcommand\Succ[1]{\operatorname{\underline{\mathsf{succ}}}(#1)}
\newcommand\Num[1]{\underline{#1}}
\newcommand\Loop\Omega
\newcommand\Loopt[1]{\Omega^{#1}}
\newcommand\Dice[1]{\underline{\operatorname{\mathsf{coin}}}(#1)}
\newcommand\Tseq[3]{{#1}\vdash{#2}:{#3}}
\newcommand\Tseqst[2]{{#1}\vdash{#2}}

\newcommand\Timpl\Impl
\newcommand\Simpl\Impl

\newcommand\PCFP{\mathsf{pPCF}}
\newcommand\PCF{\mathsf{PCF}}

\newcommand\Fun[1]{\widehat{#1}}

\newcommand\Id{\operatorname{\mathsf{Id}}}

\newcommand\Proj[1]{\pi_{#1}}

\newcommand\Excl[1]{\oc{#1}}
\newcommand\Exclp[1]{\oc({#1})}

\newcommand\Prom[1]{#1^!}

\newcommand\Relincl\eta
\newcommand\Relrestr\rho

\newcommand\Compl{\,}

\newcommand\Kl[1]{{#1}_\oc}

\newcommand\IF{\mathsf{if}}

\newcommand\Ifv[4]{\IF(#1,#2,#3\cdot #4)}

\newcommand\APP[2]{(#1)#2}

\newcommand\Eval[2]{\langle#1,#2\rangle}

\newcommand\Let[3]{\underline{\mathsf{let}}(#1,#2,#3)}

\newcommand\Snum[1]{\Base{#1}}

\newcommand\Ssuc{{\mathsf{suc}}}
\newcommand\Spred{{\mathsf{pred}}}

\newcommand\Vect[1]{\vec{#1}}

\newcommand\Obseq{\sim}

\newcommand\Bnfeq{\mathrel{\mathord:\mathord=}}
\newcommand\Bnfor{\,\,\mathord|\,\,}

\newcommand\PPCF{\mathsf{pPCF}}

\newcommand\Locpcs[2]{{#1}_{#2}}

\newcommand\Deriv[1]{{#1}'}

\newcommand\Klfun[1]{\Fun{#1}}

\newcommand\Distsp[3]{\mathsf d_{#1}(#2,#3)}
\newcommand\Distspsymb[1]{\mathsf d_{#1}}
\newcommand\Distobs[2]{\mathsf d_{\mathsf{obs}}(#1,#2)}

\newcommand\Cuball[1]{\mathcal B#1}
\newcommand\Cuballp[1]{\mathcal B(#1)}

\newcommand\Cloc[2]{{#1}_{#2}}

\newcommand\Rem[1]{\widetilde{#1}}

\newcommand\Probared[2]{\mathbb{P}(#1\downarrow #2)}

\newcommand\Mark[2]{{#1}^{#2}}

\newcommand\Starg[1]{\mathsf{arg}(#1)}
\newcommand\Stsucc{\mathsf{succ}}
\newcommand\Stpred{\mathsf{pred}}
\newcommand\Stif[2]{\mathsf{if}(#1,#2)}

\newcommand\Stlet[2]{\mathsf{let}(#1,#2)}
\newcommand\State[2]{\langle#1,#2\rangle}

\newcommand\Stcons{\cdot}
\newcommand\Stempty{\epsilon}
\newcommand\Oempty{\Seqempty}

\newcommand\Labels{\cL}

\newcommand\Ocons[2]{\Tuple{#1} #2}
\newcommand\Osingle[1]{\langle#1\rangle}

\newcommand\Inistate[1]{\State{#1}{\Stempty}}

\newcommand\Len[1]{\mathsf{len}(#1)}

\newcommand\Pcfst{\mathsf S}
\newcommand\Evalst{\mathsf{Ev}}
\newcommand\Evdom[1]{\cD(#1)}

\newcommand{\Labext}{\mathsf{lab}}
\newcommand\PPCFlab{\PPCF_\Labext}
\newcommand\PCFlab{\PCF_\Labext}
\newcommand\Clabext{\mathsf{lc}}
\newcommand\PPCFlc{\PPCF_\Clabext}

\newcommand\Cantor{\cC}
\newcommand\Cantorfin{\Cantor_0}

\newcommand\Evalstlab{\Evalst_\Labext}
\newcommand\Evalstlc{\Evalst_\Clabext}
\newcommand\Evalstlcsh{\Evalst_\Clabext^\eta}
\newcommand\Evalstlcshinv{\Evalst_\Clabext^{-\eta}}
\newcommand\Pcfstlab{\Pcfst_\Labext}
\newcommand\Pcfstlc{\Pcfst_\Clabext}
\newcommand\Striplab[1]{\mathsf{strip}(#1)}
\newcommand\Striplc[1]{\mathsf{strip}(#1)}
\newcommand\Evdomlab[1]{\cD_\Labext(#1)}
\newcommand\Evdomlc[1]{\cD_\Clabext(#1)}

\newcommand\Dicelab[2]{\operatorname{\mathsf{lcoin}}(#1,#2)}

\newcommand\Labof[1]{\mathsf{lab}(#1)}
\newcommand\Lcof[1]{\mathsf{lc}_{#1}}

\newcommand\Evalstf[1]{\mathsf{Ev}(#1)}
\newcommand\Evalstlabf[1]{\Evalstlab(#1)}
\newcommand\Evalstlabfp[2]{\Evalstlab(#1)_{#2}}

\newcommand\Eventconvz[1]{{#1}\downarrow\Num 0}

\newcommand\Proba[1]{\mathbb P_{#1}}
\newcommand\Expect[1]{\mathbb E(#1)}

\newcommand\Spy[1]{\mathsf{sp}_{#1}}

\newcommand\Tamedc[2]{{#1}^{\langle#2\rangle}}
\newcommand\Tdistobs[3]{\mathsf d^{\langle#1\rangle}_{\mathsf{obs}}(#2,#3)}
\newcommand\Tdistobsf[1]{\mathsf d^{\langle#1\rangle}}

\newcommand\Ldistobs[2]{\mathsf d_{\mathsf{lin}}(#1,#2)}

\newcommand\Undef{\mathord\uparrow}
\newcommand\Eset[1]{\{#1\}}

\newcommand\Pneg[2]{\nu_{#1}(#2)}

\newcommand\Zeroc[1]{0_{#1}}
\newcommand\Underc[1]{\underline{#1}}
\newcommand\Sseq[4]{#1\vdash #2:#3:#4}

\newenvironment{theorem}{\begin{thm}}{\end{thm}}
\newenvironment{lemma}{\begin{lem}}{\end{lem}}

\newenvironment{proposition}{\begin{prop}}{\end{prop}}

\newenvironment{remark}{\begin{rem}}{\end{rem}}
\newenvironment{example}{\begin{exa}}{\end{exa}}

\usepackage[mathlines]{lineno}

\begin{document}

\title[Differentials in $\PCOH$]
{Differentials and distances in probabilistic coherence spaces}

\author[T.~Ehrhard]{Thomas Ehrhard}
\address{Université de Paris, IRIF, CNRS, F-75013 Paris, France}
  \email{ehrhard@irif.fr}
  \thanks{This work has been partly funded by the ANR PRC project \emph{Probabilistic Programming Semanics} (PPS) ANR-19-CE48-0014.}

\begin{abstract}
  In probabilistic coherence spaces, a denotational model of
  probabilistic functional languages, morphisms are analytic and
  therefore smooth. We explore two related applications of the
  corresponding derivatives. First we show how derivatives allow to
  compute the expectation of execution time in the weak head reduction
  of probabilistic PCF (pPCF). Next we apply a general notion of
  ``local'' differential of morphisms to the proof of a Lipschitz
  property of these morphisms allowing in turn to relate the
  observational distance on pPCF terms to a distance the model is
  naturally equipped with. This suggests that extending probabilistic
  programming languages with derivatives, in the spirit of the
  differential lambda-calculus, could be quite meaningful.
\end{abstract}

\maketitle

\section*{Introduction}
Currently available denotational models of probabilistic functional
programming (with full recursion, and thus partial computations) can
be divided in three classes.
\begin{itemize}
\item \emph{Game} based models, first proposed
  in~\cite{DanosHarmer00} and further developed by various authors
  (see~\cite{CastellanClairambaultPaquetWinskel18} for an example of
  this approach). From their deterministic ancestors they typically
  inherit good definability features.
\item Models based on Scott continuous functions on domains endowed
  with additional probability related structures. Among these models
  we can mention 
  \emph{Kegelspitzen}~\cite{KeimelPlotkin17} (domains equipped with an
  algebraic convex structure) and \emph{$\omega$-quasi Borel
    spaces}~\cite{VakarKammarStaton19} (domains equipped with a
  generalized notion of measurability).
\item Models based on (a generalization of) Berry stable
  functions. The first category of this kind was that of
  \emph{probabilistic coherence spaces} (PCSs) and power series with
  non-negative coefficients (the Kleisli category of the model of
  Linear Logic developed in~\cite{DanosEhrhard08}) for which we could
  prove adequacy and full abstraction with respect to a probabilistic
  version of $\PCF$~\cite{EhrhardPaganiTasson18,Ehrhard20}. We
  extended this idea to ``continuous data types'' (such as $\Real$) by
  substituting PCSs with \emph{positive cones} and power series with
  functions featuring an hereditary monotonicity property that we
  called~\emph{stability}\footnote{Because, when reformulated in the
    domain-theoretic framework of Girard's coherence spaces, this
    condition exactly characterizes Berry's stable functions.}
  and~\cite{Crubille18} showed that this extension is actually
  conservative (stable functions on PCSs, which are special positive
  cones, are exactly power series).
\end{itemize}

The main feature of this latter semantics is the extreme regularity of
its morphisms. Being power series, they must be smooth. Nevertheless,
the category $\PCOH$ is not a model of differential linear logic in
the sense of~\cite{Ehrhard18}. This is due to the fact that general
addition of morphisms is not possible (only sub-convex linear
combinations are available) thus preventing, \Eg, the Leibniz rule to
hold in the way it is presented in differential LL. Also a morphism
$X\to Y$ in the Kleisli category $\Kl\PCOH$ can be considered as a
function from the \emph{closed unit ball} of the cone $P$ associated
with $X$ to the closed unit ball of the cone $Q$ associated with
$Y$. From a differential point of view such a morphism is well behaved
only in the interior of the unit ball. On the border derivatives can
typically take infinite values.

\paragraph*{Contents}
We already used the analyticity of the morphisms of $\Kl\PCOH$ to
prove full abstraction results~\cite{EhrhardPaganiTasson18}. We
provide here two more corollaries of this property, involving
also derivatives. For both results, we consider a paradigmatic
probabilistic purely functional programming language\footnote{One
  distinctive feature of our approach is to not consider probabilities
  as an effect.} which is a probabilistic extension of Scott and
Plotkin's PCF. This language $\PPCF$ features a single data type
$\Tnat$ of integers, a simple probabilistic choice operator
$\Dice r:\Tnat$ which flips a coin with probability $r$ to get
$\Num 0$ and $1-r$ to get $\Num 1$. To make probabilistic programming
possible, this language has a $\Let xMN$ construct restricted to $M$
of type $\Tnat$ which allows to sample an integer according to the
sub-probability distribution represented by $M$. The operational
semantics is presented by a deterministic ``stack machine'' which is
an environment-free Krivine machine
parameterized by a choice sequence $\in\Cantorfin=\{0,1\}^{<\omega}$,
presented as a partial \emph{evaluation function}. We adopt a standard
discrete probability approach, considering $\Cantorfin$ as our basic
sample space and the
evaluation function as defining a (total) probability density function
on $\Cantorfin$. We also introduce an extension $\PPCFlab$ of $\PPCF$
where terms can be labeled by elements of a set $\Labels$ of labels,
making it possible to count the use of labeled subterms of a term $M$
(closed and of ground type) during a reduction of $M$. Evaluation for
this extended calculus gives rise to a random variable (r.v.) on
$\Cantorfin$ ranging in the set $\Mfin\cL$ of finite multisets of
elements of $\Labels$. The number of uses of terms labeled by a given
$l\in\Labels$ (which is a measure of the computation time) is then an
$\Nat$-valued r.v., the expectation of which we want to evaluate.  We
prove that, for a given labeled closed term $M$ of type $\Tnat$, this
expectation can be computed by taking a derivative of the
interpretation of this term in the model $\Kl\PCOH$ and provide a
concrete example of computation of such expectations. This result can
be considered as a probabilistic version
of~\cite{DeCarvalho09,DeCarvalho18}. The fact that derivatives can
become infinite on the border of the unit ball corresponds then to the
fact that this expectation of ``computation time'' can be infinite.

In the second application, we consider the contextual distance on
$\PPCF$ terms generalizing Morris equivalence as studied
in~\cite{CrubilleDalLago17} for instance.  The probabilistic features
of the language make this distance too discriminating, putting
\Eg~terms $\Dice 0$ and $\Dice\epsilon$ at distance $1$ for all
$\epsilon>0$ (\emph{probability amplification}). Any cone (and hence
any PCS) is equipped with a norm and hence a canonically defined
metric\footnote{See Remark~\ref{rk:cone-distance} for the definition
  of this distance for general cones.}. Using a \emph{locally defined}
notion of differential of morphisms in $\Kl\PCOH$, we prove that these
morphisms enjoy a Lipschitz property on all balls of radius $p<1$,
with a Lipschitz constant $1/(1-p)$ (thus tending towards $\infty$
when $p$ tends towards $1$). Modifying the definition of the
operational distance by not considering all possible contexts, but
only those which ``perturb'' the tested terms by allowing them to
diverge with probability $1-p$, we upper bound this $p$-tamed distance
by the distance of the model with a ratio $p/(1-p)$. Being in some
sense defined wrt.~\emph{linear} semantic contexts, the denotational
distance does not suffer from the probability amplification
phenomenon. This suggests that $p$-tamed distances might be more
suitable than ordinary contextual distances to reason on probabilistic
programs.

\paragraph*{Notations}
We use $\Realp$ for the set of real numbers $x$ such that $x\geq 0$,
and we set $\Realpc=\Realp\cup\{+\infty\}$.  Given two sets $S$ and
$I$ we use $S^I$ for the set of functions $I\to S$, often considered
as $I$-indexed families $\Vect s$ of elements of $S$. We use the
notation $\Vect s$ (with an arrow) when we want to stress the fact
that the considered object is considered as an indexed family, the
indexing set $I$ being usually easily derivable from the context. The
elements of such a family $\Vect s$ are denoted $s_i$ or $s(i)$
depending on the context (to avoid accumulations of subscripts).
Given $i\in I$ we use $\Base i$ for the function $I\to\Realp$ such
that $\Base i(i)=1$ and $\Base i(j)=0$ if $j\not=i$. In other words
$\Base i(j)=\Kronecker ij$, the Kronecker symbol.  We use $\Mfin I$
for the set of finite multisets of elements of $I$. A multiset is
a function $\mu:I\to\Nat$ such that
$\Supp\mu=\{i\in I\St\mu(i)\not=0\}$ is finite. We use additive
notations for operations on multisets ($0$ for the empty multiset,
$\mu+\nu$ for their pointwise sum). We use $\Mset{\List i1k}$ for the
multiset $\mu$ such that $\mu(i)=\Card{\{j\in\Nat\St i_j=i\}}$. If
$\mu,\nu\in\Mfin I$ with $\mu\leq\nu$ (pointwise order), we set
$\Binom\nu\mu=\prod_{i\in I}\Binom{\nu(i)}{\mu(i)}$ where
$\Binom nm=\frac{\Factor n}{\Factor m\Factor{(n-m)}}$ is the usual
binomial coefficient. Given $\mu\in\Mfin I$ and $i\in I$ we write
$i\in\mu$ if $\mu(i)\not=0$ and we set
$\Supp\mu=\Eset{i\in I\St i\in\mu}$.

We use $I^{<\omega}$ for the set of finite sequences
$\Tuple{\List i1k}$ of elements of $I$ and $\alpha\,\beta$ for the
concatenation of such sequences. We use $\Seqempty$ for the empty
sequence.

\section{Probabilistic coherence spaces (PCS)}

For the general theory of PCSs we refer
to~\cite{DanosEhrhard08,EhrhardPaganiTasson18} where the reader will
find a more detailed presentation, including motivating
examples. Here, we recall only the basic definitions and provide a
characterization of these objects. So this section should not be
considered as an introduction to PCSs: for such an introduction the
reader is advised to have a look at the articles mentioned above. PCSs
are particular \emph{positive cones}, a notion borrowed
from~\cite{Selinger04}) that we used in~\cite{EhrhardPaganiTasson18}
to extend the probabilistic semantics of PCS to continuous data-types
such as the real line.

\subsection{A few words about cones}

A (positive) \emph{pre-cone} is a cancellative\footnote{Meaning that
  $x+y=x'+y\Implies x=x'$.} commutative $\Realp$-semi-module $P$
equipped with a norm $\Normsp\_ P$, that is a map $P\to\Realp$, such
that $\Normsp {r\,x}P=r\,\Normsp xP$ for $r\in\Realp$,
$\Normsp{x+y}P\leq\Normsp xP+\Normsp yP$ and
$\Normsp xP=0\Implies x=0$. It is moreover assumed that
$\Normsp xP\leq\Normsp{x+y}P$, this condition expressing that the
elements of $P$ are positive. Given $x,y\in P$, one says that $x$ is
less than $y$ (notation $x\leq y$) if there exists $z\in P$ such that
$x+z=y$. By the cancellativeness property, if such a $z$ exists, it is
unique and we denote it as $y-x$. This subtraction obeys usual
algebraic laws (when it is defined).  Notice that if $x,y\in P$
satisfy $x+y=0$ then since $\Normsp xP\leq\Normsp{x+y}P$, we have
$x=0$ (and of course also $y=0$). Therefore, if $x\leq y$ and
$y\leq x$ then $x=y$ and so $\leq$ is an order relation.

A (positive) \emph{cone} is a positive pre-cone $P$ whose unit ball
$\Cuball P=\{x\in P\St\Normsp xP\leq 1\}$ is $\omega$-order-complete
in the sense that any increasing sequence of elements of $\Cuball P$
has a least upper bound in $\Cuball
P$. In~\cite{EhrhardPaganiTasson18} we show how a notion of
\emph{stable} function on cones can be defined, which gives rise to a
cartesian closed category and in~\cite{Ehrhard20} we explore the
category of cones and linear and Scott-continuous functions.

\subsection{Basic definitions on PCSs}\label{sec:basics-PCSs}

Given an at most countable set $I$ and $u,u'\in\Realpcto I$, we set
$\Eval u{u'}=\sum_{i\in I}u_iu'_i\in\Realpc$. Given
$P\subseteq\Realpcto I$, we define $\Orth P\subseteq\Realpcto I$ as
\begin{align*}
  \Orth P=\{u'\in\Realpcto I\St\forall u\in P\ \Eval u{u'}\leq 1\}\,.
\end{align*}
Observe that if $P$ satisfies
\( \forall a\in I\,\exists x\in P\ x_a>0 \) and
\( \forall a\in I\,\exists m\in\Realp \forall x\in P\ x_a\leq m \)
then $\Orth P\in\Realpto I$ and $\Orth P$ satisfies the same two
properties.

A probabilistic pre-coherence space (pre-PCS) is a pair
$X=(\Web X,\Pcoh X)$ where $\Web X$ is an at most countable
set\footnote{This restriction is not technically necessary, but very
  meaningful from a philosophic point of view; the non countable case
  should be handled via measurable spaces and then one has to consider
  more general objects as in~\cite{EhrhardPaganiTasson18} for
  instance.} and $\Pcoh X\subseteq\Realpcto{\Web X}$ satisfies
$\Biorth{\Pcoh X}=\Pcoh X$. A probabilistic coherence space (PCS) is a
pre-PCS $X$ such that
\(
\forall a\in\Web X\,\exists x\in\Pcoh X\ x_a>0
\) and
\(
\forall a\in\Web X\,\exists m\in\Realp \forall x\in\Pcoh X\ x_a\leq m
\)
or equivalently
\begin{align*}
  \forall a\in\Web X\quad0<\sup_{x\in\Pcoh X}x_a<\infty
\end{align*}
so that $\Pcoh X\subseteq\Realpto{\Web X}$.

Given any PCS $X$ we can define a cone $\Pcohc X$ as follows:
\begin{linenomath}
\begin{align*}
  \Pcohc X=\{x\in\Realpto{\Web X}\St
  \exists\epsilon>0\ \epsilon x\in\Pcoh X\}
\end{align*}
\end{linenomath}
that we equip with the following norm:
\( \Normsp x{\Pcohc X}=\inf\{r>0\St x\in r\,\Pcoh X\} \) and then it
is easy to check that $\Cuballp{\Pcohc X}=\Pcoh X$. We simply denote
this norm as $\Norm\__X$, so that
$\Norm x_X=\sup_{x'\in\Pcoh{\Orth X}}\Eval x{x'}$.

Given $t\in\Realpcto{I\times J}$ considered as a matrix (where $I$ and
$J$ are at most countable sets) and $u\in\Realpcto I$, we define
$\Matappa tu\in\Realpcto J$ by $(\Matappa tu)_j=\sum_{i\in I}t_{i,j}u_i$
(usual formula for applying a matrix to a vector), and if
$s\in\Realpcto{J\times K}$ we define the product
$\Matapp st\in\Realpcto{I\times K}$ of the matrix $s$ and $t$ as usual
by $(\Matapp st)_{i,k}=\sum_{j\in J}t_{i,j}s_{j,k}$. This is an
associative operation.

Let $X$ and $Y$ be PCSs, a morphism from $X$ to $Y$ is a matrix
$t\in\Realpto{\Web X\times\Web Y}$ such that
$\forall x\in\Pcoh X\ \Matappa tx\in\Pcoh Y$. It is clear that the
identity matrix is a morphism from $X$ to $X$ and that the matrix
product of two morphisms is a morphism and therefore, PCSs equipped
with this notion of morphism form a category $\PCOH$.

The condition $t\in\PCOH(X,Y)$ is equivalent to
\(
\forall x\in\Pcoh X\,\forall y'\in\Pcoh{\Orth Y}\ \Eval{\Matappa tx}{y'}\leq 1
\)
but $\Eval{\Matappa tx}{y'}=\Eval t{\Tens x{y'}}$ where
$(\Tens x{y'})_{(a,b)}=x_ay'_b$. We define
$\Limpl XY=(\Web X\times\Web Y,\{t\in\Realpto{\Web{\Limpl
    XY}}\St\forall x\in\Pcoh X\ \Matappa tx\in\Pcoh Y\})$: this is a
pre-PCS by this observation, and checking that it is indeed a PCS is
easy.

We define then $\Tens XY=\Orthp{\Limpl X{\Orth Y}}$; this is a PCS which satisfies
\(
\Pcohp{\Tens XZ}=\Biorth{\{\Tens xz\St x\in\Pcoh X\text{ and }z\in\Pcoh Z\}}
\)
where $(\Tens xz)_{(a,c)}=x_az_c$. Then it is easy to see that we have
equipped in that way the category $\PCOH$ with a symmetric monoidal
structure for which it is $\ast$-autonomous wrt.~the dualizing object
$\Bot=\One=(\{*\},[0,1])$ which is also the unit of $\ITens$. The
$\ast$-autonomy follows easily from the observation that
$(\Limpl X\Bot)\Isom\Orth P$.

The category $\PCOH$ is cartesian: if $(X_i)_{i\in I}$ is an at most
countable family of PCSs, then
$(\Bwith_{i\in I}X_i,(\Proj i)_{i\in I})$ is the cartesian product of
the $X_i$s, with
$\Web{\Bwith_{i\in I}X_i}=\Bunion_{i\in I}\{i\}\times\Web{X_i}$,
$(\Proj i)_{(j,a),a'}=1$ if $i=j$ and $a=a'$ and
$(\Proj i)_{(j,a),a'}=0$ otherwise, and
$x\in\Pcohp{\Bwith_{i\in I}X_i}$ if $\Matappa{\Proj i}x\in\Pcoh{X_i}$
for each $i\in I$ (for $x\in\Realpto{\Web{\Bwith_{i\in
      I}X_i}}$). Given $t_i\in\PCOH(Y,X_i)$, the unique morphism
$t=\Tuple{t_i}_{i\in I}\in\PCOH(Y,\Bwith_{i\in I}X_i)$ such that
$\Proj i\Compl t=t_i$ is simply defined by
$t_{b,(i,a)}=(t_i)_{a,b}$. The dual operation $\Bplus_{i\in I}X_i$,
which is a coproduct, is characterized by
$\Web{\Bplus_{i\in I}X_i}=\Bunion_{i\in I}\{i\}\times\Web{X_i}$ and
$x\in\Pcohp{\Bplus_{i\in I}X_i}$ and
$\sum_{i\in I}\Norm{\Proj i\Compl x}_{X_i}\leq 1$.

A particular case is $\Snat=\Bplus_{n\in\Nat}X_n$ where $X_n=\One$ for
each $n$. So that $\Web\Snat=\Nat$ and $x\in\Realpto\Nat$ belongs to
$\Pcoh\Snat$ if $\sum_{n\in\Nat}x_n\leq 1$ (that is, $x$ is a
sub-probability distribution on $\Nat$). For each $n\in\Nat$ we have
$\Base n\in\Pcoh\Snat$ which is the distribution concentrated
on the integer $n$.  There are successor and predecessor morphisms
$\Ssuc,\Spred\in\PCOH(\Snat,\Snat)$ given by
$\Ssuc_{n,n'}=\Kronecker{n+1}{n'}$ and $\Spred_{n,n'}=1$ if $n=n'=0$
or $n=n'+1$ (and $\Spred_{n,n'}=0$ in all other cases). An element of
$\PCOH(\Snat,\Snat)$ is a (sub)stochastic matrix and the very idea of
this model is to represent programs as transformations of this kind,
and their generalizations.

As to the exponentials, one sets $\Web{\Excl X}=\Mfin{\Web X}$ and
$\Pcohp{\Excl X}=\Biorth{\{\Prom x\St x\in\Pcoh X\}}$ where, given
$\mu\in\Mfin{\Web X}$, $\Prom x_\mu=x^\mu=\prod_{a\in\Web
  X}x_a^{\mu(a)}$. Then given $t\in\PCOH(X,Y)$, one defines $\Excl
t\in\PCOH(\Excl X,\Excl Y)$ in such a way that $\Matappa{\Excl t}{\Prom
  x}=\Prom{(\Matappa tx)}$ (the precise definition is not relevant
here; it is completely determined by this equation). We do not need
here to specify the monoidal comonad structure of this
exponential. The resulting cartesian closed category\footnote{This is
  the Kleisli category of ``$\oc$'' which has actually a comonad
  structure that we do not make explicit here, again we refer
  to~\cite{DanosEhrhard08,EhrhardPaganiTasson18}.}  $\Kl\PCOH$ can be
seen as a category of functions (actually, of stable functions as
proved in~\cite{Crubille18}). Indeed, a morphism
$t\in\Kl\PCOH(X,Y)=\PCOH(\Excl X,Y)=\Pcohp{\Limpl{\Excl X}{Y}}$ is
completely characterized by the associated function $\Fun t:\Pcoh
X\to\Pcoh Y$ such that $\Fun t(x)=\Matappa t{\Prom
  x}=\left(\sum_{\mu\in\Web{\Excl X}}t_{\mu,b}x^\mu\right)_{b\in\Web
  Y}$ so that we consider morphisms as power series (they are in
particular monotonic and Scott continuous functions $\Pcoh X\to\Pcoh
Y$). In this cartesian closed category, the product of a family
$(X_i)_{i\in I}$ is $\Bwith_{i\in I}X_i$ (written $X^I$ if $X_i=X$ for
all $i$), which is compatible with our viewpoint on morphisms as
functions since $\Pcohp{\Bwith_{i\in I}X_i}=\prod_{i\in I}\Pcoh{X_i}$
up to trivial iso. The object of morphisms from $X$ to $Y$ is
$\Limpl{\Excl X}{Y}$ with evaluation mapping
$(t,x)\in\Pcohp{\Limpl{\Excl X}{Y}}\times\Pcoh X$ to $\Fun t(x)$ that
we simply denote as $t(x)$ from now on. The well defined function
$\Pcohp{\Limpl{\Excl X}X}\to\Pcoh X$ which maps $t$ to
$\sup_{n\in\Nat}t^n(0)$ is a morphism of $\Kl\PCOH$ (and thus can be
described as a power series in the vector $t=(t_{m,a})_{m\in\Mfin{\Web
    X},a\in\Web X}$) by standard categorical considerations using
cartesian closeness: it provides us with fixed point operators at all
types.

\section{Probabilistic PCF, time expectation and
  derivatives}\label{sec:PPCF-derivatives}
We introduce now the probabilistic functional programming language
considered in this paper. The operational semantics is presented using
elementary probability theoretic tools.

\subsection{The core language}\label{sec:PPCF-core}
The types and terms are given by
\begin{linenomath}
\begin{align*}
  \sigma,\tau,\dots
  &\Bnfeq \Tnat \Bnfor \Timpl\sigma\tau\\
  M,N,P\dots
  &\Bnfeq \Num n \Bnfor \Succ M
    \Bnfor \Pred M \Bnfor x \Bnfor \Dice r
    \Bnfor \Let xMN \Bnfor \If MNP\\
  &\quad\quad\Bnfor \App MN \Bnfor \Abst x\sigma M
    \Bnfor \Fix M
\end{align*}
\end{linenomath}

See Fig.~\ref{fig:pPCF-typing} for the typing rules, with typing
contexts $\Gamma=(x_1:\sigma_1,\dots,x_n:\sigma_n)$; notice that this
figures includes the typing rules for the stacks that we introduce
below. It is important to keep in mind that it would not make sense to
extend the construction $\Let zMN$ to terms $M$ which are not of type
$\Tnat$. This construction uses essentially the fact that the type
$\Tnat$ is a \emph{positive} formula of linear logic,
see~\cite{EhrhardTasson16}.
%
\begin{figure}
\footnotesize{
\begin{center}
  \AxiomC{}
  \UnaryInfC{$\Tseq\Gamma{\Num n}\Tnat$}
  \DisplayProof
  \quad
  \AxiomC{}
  \UnaryInfC{$\Tseq{\Gamma,x:\sigma}{x}\sigma$}
  \DisplayProof
  \quad
  \AxiomC{$\Tseq\Gamma M\Tnat$}
  \UnaryInfC{$\Tseq\Gamma{\Succ M}\Tnat$}
  \DisplayProof
  \quad
  \AxiomC{$\Tseq\Gamma M\Tnat$}
  \UnaryInfC{$\Tseq\Gamma{\Pred M}\Tnat$}
  \DisplayProof
\end{center}
\begin{center}
  \AxiomC{$\Tseq\Gamma M\Tnat$}
  \AxiomC{$\Tseq\Gamma N\sigma$}
  \AxiomC{$\Tseq\Gamma P\sigma$}
  \TrinaryInfC{$\Tseq\Gamma{\If MNP}\sigma$}
  \DisplayProof
  \quad
  \AxiomC{$\Tseq\Gamma M\Tnat$}
  \AxiomC{$\Tseq{\Gamma,z:\Tnat}{N}\sigma$}
  \BinaryInfC{$\Tseq\Gamma{\Let zMN}\sigma$}
  \DisplayProof
\end{center}
\begin{center}
  \AxiomC{$\Tseq{\Gamma,x:\sigma}M\tau$}
  \UnaryInfC{$\Tseq\Gamma{\Abst x\sigma M}{\Timpl\sigma\tau}$}
  \DisplayProof
  \quad
  \AxiomC{$\Tseq\Gamma M{\Timpl\sigma\tau}$}
  \AxiomC{$\Tseq\Gamma N\sigma$}
  \BinaryInfC{$\Tseq\Gamma{\App MN}\tau$}
  \DisplayProof
  \quad
  \AxiomC{$\Tseq\Gamma M{\Timpl\sigma\sigma}$}
  \UnaryInfC{$\Tseq\Gamma{\Fix M}\sigma$}
  \DisplayProof
  \quad
  \AxiomC{$r\in[0,1]\cap\Rational$}
  \UnaryInfC{$\Tseq\Gamma{\Dice r}\Tnat$}
  \DisplayProof
\end{center}
\begin{center}
  \AxiomC{}
  \UnaryInfC{$\Tseqst\Tnat{\Stempty}$}
  \DisplayProof
  \quad
  \AxiomC{$\Tseq {}M\sigma$}
  \AxiomC{$\Tseqst\tau\pi$}
  \BinaryInfC{$\Tseqst{\Timpl\sigma\tau}{{\Starg M}\Stcons\pi}$}
  \DisplayProof
  \quad
  \AxiomC{$\Tseqst\Tnat\pi$}
  \UnaryInfC{$\Tseqst\Tnat{\Stsucc\Stcons\pi}$}
  \DisplayProof
  \quad
  \AxiomC{$\Tseqst\Tnat\pi$}
  \UnaryInfC{$\Tseqst\Tnat{\Stpred\Stcons\pi}$}
  \DisplayProof
\end{center}
\begin{center}
  \AxiomC{$\Tseq{}N\sigma$}
  \AxiomC{$\Tseq{}P\sigma$}
  \AxiomC{$\Tseqst\sigma\pi$}
  \TrinaryInfC{$\Tseqst\Tnat{\Stif NP\Stcons\pi}$}
  \DisplayProof
  \quad
  \AxiomC{$\Tseq{x:\Tnat}{N}{\sigma}$}
  \AxiomC{$\Tseqst\sigma\pi$}
  \BinaryInfC{$\Tseqst\Tnat{{\Stlet xN}\Stcons\pi}$}
  \DisplayProof
\end{center}}
  \caption{Typing rules for $\PPCF$ terms and stacks}
  \label{fig:pPCF-typing}
\end{figure}

\subsubsection{Denotational semantics}\label{sec:pPCF-den-sem-pcoh}
We survey briefly the interpretation of $\PPCF$ in PCSs thoroughly
described in~\cite{EhrhardPaganiTasson18}. Types are interpreted by
$\Tsem\Tnat=\Snat$ and
$\Tsem{\Timpl\sigma\tau}=\Limpl{\Excl{\Tsem\sigma}}{\Tsem\tau}$. Given
$M\in\PPCF$ such that $\Tseq\Gamma M\sigma$ (with
$\Gamma=(x_1:\sigma_1,\dots,x_k:\sigma_k)$) one defines
$\Psem M\Gamma\in\Kl\PCOH(\Bwith_{i=1}^k\Tsem{\sigma_i},\Tsem\sigma)$
(a ``Kleisli morphism'') that we see as a function
$\prod_{i=1}^k\Pcoh{\Tsem{\sigma_i}}\to\Pcoh{\Tsem\sigma}$ as
explained in Section~\ref{sec:basics-PCSs}.
These functions are given by
\begin{linenomath}
  \begin{align*}
    \Psem{\Num n}\Gamma(\Vect u)
    &=\Snum n\\
    \Psem{x_i}\Gamma(\Vect u)
    &=u_i\\
    \Psem{\Dice r}\Gamma(\Vect u)
    &=r\,\Snum 0+(1-r)\,\Snum 1\\
    \Psem{\Succ M}\Gamma(\Vect u)
    &=\Matappa\Ssuc{\Psem M\Gamma(\Vect u)}
      =\sum_{n\in\Nat}\Psem M\Gamma(\Vect u)_n\Snum{n+1}\\
    \Psem{\Pred M}\Gamma(\Vect u)
    &=\Matappa\Spred{\Psem M\Gamma(\Vect u)}
      =\Psem M\Gamma(\Vect u)_{0}\Snum{0}
      +\sum_{n\in\Nat}\Psem M\Gamma(\Vect u)_{n+1}\Snum{n}\\
    \Psem{\Let xMN}\Gamma(\Vect u)
    &=\sum_{n\in\Nat}\Psem M\Gamma(\Vect
      u)_n\,\Psem{\Subst N{\Num n}x}\Gamma(\Vect u)\\
    \Psem{\If MNP}\Gamma(\Vect u)
    &=\Psem M\Gamma(\Vect u)_0\,\Psem
      N\Gamma(\Vect u)+\left(\sum_{n\in\Nat}
      \Psem M\Gamma(\Vect u)_{n+1}\right)\Psem
      P\Gamma(\Vect u)\\
    \Psem{\App MN}\Gamma(\Vect u)
    &=(\Psem M\Gamma(\Vect u))(\Psem N\Gamma(\Vect u))\\
    \Psem{\Fix M}\Gamma(\Vect u)
    &=\sup_{n\in\Nat}(\Psem M\Gamma(\Vect u))^n(0)
\end{align*}
\end{linenomath}
and, assuming that $\Tseq{\Gamma,x:\sigma}M\tau$ and
$\Vect u\in\prod_{i=1}^k\Pcoh{\Tsem{\sigma_i}}$,
$\Psem{\Abst x\sigma M}\Gamma(\Vect u)$ is the element $t$ of
$\Pcoh{(\Limpl{\Excl{\Tsem\sigma}}{\Tsem\tau})}$ characterized by
$\forall u\in\Pcoh{\Tsem\sigma}\ \Fun t(u)=\Psem
M{\Gamma,x:\sigma}(\Vect u,u)$.

\subsubsection{Operational semantics}\label{sec:PCF-operational}
In former papers we have presented the operational semantics of
$\PPCF$ as a discrete Markov chain on states which are the closed
terms of $\PPCF$. This Markov chain implements the standard weak head
reduction strategy of PCF which is deterministic for ordinary PCF but
features branching in $\PPCF$ because of the $\Dice r$ construct
(see~\cite{EhrhardPaganiTasson18}). Here we prefer another, though
strictly equivalent, presentation of this operational semantics, based
on an environment-free Krivine Machine (thus handling states which are
pairs made of a closed term and a closed stack) further parameterized
by an element of $\{0,1\}^{<\omega}$ to be understood as a ``random
tape'' prescribing the values taken by the $\Dice r$ terms during the
execution of states. We present this machine as a partial function
taking a state $s$, a random tape $\alpha$ and returning an element of
$[0,1]$ to be understood as the probability that the sequence $\alpha$
of $0$/$1$ choices occurs during the execution of $s$. We allow only
execution of ground type states and accept $\Num 0$ as the only
terminating value: a completely arbitrary choice, sufficient for
our purpose in this paper. Also, we insist that a terminating
computation from $(s,\alpha)$ completely consumes the random tape
$\alpha$. These choices allow to fit within a completely standard
discrete probability setting.

Given an extension $\Lambda$ of $\PPCF$ (with the same format
for typing rules), we define the associated language of stacks (called
$\Lambda$-stacks).
\begin{linenomath}
\begin{align*}
  \pi \Bnfeq \Stempty \Bnfor {\Starg M}\Stcons\pi
  \Bnfor {\Stsucc}\Stcons\pi \Bnfor {\Stpred}\Stcons\pi
  \Bnfor {\Stif NP}\Stcons\pi 
  \Bnfor {\Stlet xN}\Stcons\pi
\end{align*}
\end{linenomath}
where $M$ and $N$ range over $\Lambda$.  A stack typing judgment is of
shape $\Tseqst{\sigma}{\pi}$ (meaning that the stack $\pi$ takes a
term of type $\sigma$ and returns an integer) and the typing rules are
given in Fig.~\ref{fig:pPCF-typing}.

A \emph{state} is a pair $\State{M}{\pi}$ (where we say that $M$ is
\emph{in head position}) such that $\Tseq{}M\sigma$ and
$\Tseqst\sigma\pi$ for some (uniquely determined) type $\sigma$, let
$\Pcfst$ be the set of states.  Let $\Cantorfin=\{0,1\}^{<\omega}$ be
the set of finite lists of booleans (random tapes), we define a
\emph{partial} function $\Evalst:\Pcfst\times\Cantorfin\to\Intercc01$ in
Fig.~\ref{fig:PPCF-Krivine}\footnote{Notice that all the equations
  defining $\Evalst$ in Fig.~\ref{fig:PPCF-Krivine} are well-typed in
  the sense that if the state of the LHS of an equation is well-typed,
  so is its RHS.} where we use the functions
\begin{linenomath}
\begin{align}\label{eq:Pneg-def}
  \Pneg 0r&=r\\
  \Pneg 1r&=1-r\,.
\end{align}%
\end{linenomath}
\begin{figure} {\footnotesize \begin{alignat*}{3} &\Evalst(\State{\Let
        xMN}\pi,\alpha)=\Evalst(\State M{\Stlet xN\Stcons\pi},\alpha)
      &\quad &\Evalst(\State{\App MN}{\pi},\alpha)
      =\Evalst(\State M{\Starg N\Stcons\pi},\alpha)\\
      &\Evalst(\State{\Num n}{\Stlet xN\Stcons\pi},\alpha)
      =\Evalst(\State{\Subst N{\Num n}x}{\pi},\alpha) &\quad
      &\Evalst(\State{\Abst x\sigma M}{\Starg N\Stcons\pi},\alpha)
      =\Evalst(\State{\Subst MNx}{\pi},\alpha)\\
      &\Evalst(\State{\If MNP}{\pi}) =\Evalst(\State{M}{\Stif
        NP\Stcons\pi},\alpha) &\quad &\Evalst(\State{\Fix
        M}{\pi},\alpha)
      =\Evalst(\State{M}{\Starg{\Fix M}\Stcons\pi},\alpha)\\
      &\Evalst(\State{\Num 0}{\Stif NP\Stcons\pi},\alpha)
      =\Evalst(\State{N}{\pi},\alpha) &\quad &\Evalst(\State{\Dice
        r}{\pi},\Ocons i\alpha)
      =\Evalst(\State{\Num i}{\pi},\alpha)\cdot\Pneg ir\\
      &\Evalst(\State{\Num{n+1}}{\Stif NP\Stcons\pi},\alpha)
      =\Evalst(\State{P}{\pi},\alpha) &\quad &\Evalst(\State{\Num
        0}{\Stempty},\Oempty)=1\\
      &\Evalst(\State{\Succ M}\pi,\alpha)
      =\Evalst(\State{M}{\Stsucc\Stcons\pi},\alpha)&\quad
      &\Evalst(\State{\Pred M}\pi,\alpha)
      =\Evalst(\State{M}{\Stpred\Stcons\pi},\alpha)\\
      &\Evalst(\State{\Num n}{\Stsucc\Stcons\pi},\alpha)
      =\Evalst(\State{\Num{n+1}}{\pi},\alpha)&\quad
      &\Evalst(\State{\Num 0}{\Stpred\Stcons\pi},\alpha)
      =\Evalst(\State{\Num 0}{\pi},\alpha)\\
      &&\quad
      &\Evalst(\State{\Num{n+1}}{\Stpred\Stcons\pi},\alpha)
      =\Evalst(\State{\Num n}{\pi},\alpha)
\end{alignat*}}  
  \caption{The $\PPCF$ Krivine Machine}
  \label{fig:PPCF-Krivine}
\end{figure}%
Let $\Evdom s$ be the set of all $\alpha\in\Cantorfin$ such that
$\Evalst(s,\alpha)$ is defined.
%
When $\alpha\in\Evdom s$, the number $\Evalst(s,\alpha)\in[0,1]$ is
the probability that the random tape $\alpha$ occurs during the
execution. When all coins are fair (all the values of the parameters
$r$ are $1/2$), this probability is $2^{-\Len\alpha}$. The sum of
these (possibly infinitely many) probabilities is $\leq 1$. For
fitting within a standard probabilistic setting, we define a total
probability distribution $\Evalstf s:\Cantorfin\to[0,1]$ as follows
\begin{linenomath}
\begin{equation*}
  \Evalstf s(\alpha)=
  \begin{cases}
    \Evalst(s,\beta) & \text{if }\alpha=\Ocons 0\beta
    \text{ and }\beta\in\Evdom s\\
    1-\sum_{\beta\in\Evdom s}\Evalst(s,\beta) & \text{if }\alpha=\Tuple 1\\
    0 & \text{in all other cases}
  \end{cases}
\end{equation*}
\end{linenomath}
so that $\Tuple 1$ carries the weight of divergence. We prefer this
option rather than adding an error element to $\Cantorfin$ which would
be more natural from a programming language point of view, but less
standard from the viewpoint of probability theory. This choice is
arbitrary and has no impact on the result we prove because all the
events of interest for us will be subsets of
$\Ocons 0{\Cantorfin}\subset\Cantorfin$.

Let $\Proba s$ be the associated probability measure. We are in a discrete
setting so simply
\begin{align*}
  \Proba s(A)=\sum_{\alpha\in A}\Evalstf s(\alpha)
\end{align*}
for all $A\subseteq\Cantorfin$.
The event
\begin{linenomath}
\begin{align*}
  (\Eventconvz s)=\Ocons 0{\Evdom s}
\end{align*}
\end{linenomath}
is the set of all random tapes (up to $0$-prefixing) making $s$ reduce
to $\Num 0$. Its probability is
\begin{linenomath}
\begin{align*}
  \Proba s(\Eventconvz s)=\sum_{\beta\in\Evdom s}\Evalst(s,\beta)\,.
\end{align*}
\end{linenomath}
In the case $s=\Inistate M$ (with $\Tseq{}M\Tnat$) this probability is
\emph{exactly the same} as the probability of $M$ to reduce to
$\Num 0$ in the Markov chain setting of~\cite{EhrhardPaganiTasson18}
(see \Eg~\cite{BorgstromDalLagoGordonSzymczak16} for more details on
the connection between these two kinds of operational semantics).
So the Adequacy Theorem of~\cite{EhrhardPaganiTasson18} can be
expressed as follows.
\begin{theorem}\label{th:pcoh-adequacy}
  Let $M\in\PPCF$ with $\Tseq{}M\Tnat$. Then
  ${\Psem M{}}_0=\Proba{\Inistate M}{(\Eventconvz{\Inistate M})}$.
\end{theorem}
We use sometimes $\Probared M{\Num 0}$ as an abbreviation for
$\Proba{\Inistate M}{(\Eventconvz{\Inistate M})}$.

We shall introduce several versions of PCF in the sequel, with
associated machines.
\begin{itemize}
\item In Section~\ref{sec:ppcf-labels} we introduce $\PPCFlab$ where
  terms can be labeled by elements of $\Labels$ and the machine
  $\Evalstlab$ which returns a multiset of elements of $\Labels$
  counting how many times labeled subterms arrive in head position
  during the evaluation. 
\item In Section~\ref{sec:PCF-lc} we introduce $\PPCFlc$ which
  includes a labeled version of the $\Dice\_$ construct, and the
  associated machine $\Evalstlc$ which returns an element of
  $\Realp$ (a probability actually).
\item In the same section we introduce as an auxiliary tool the
  machine $\Evalstlcsh$ which differs from the previous one by the
  fact that it returns an element of $\Cantorfin$.
\item We also introduce a machine $\Evalstlcshinv$ which is a kind of
  inverse of the previous one.
\end{itemize}
In the proofs these machines will often be used with an additional
integer parameter for indexing the execution steps.

It is important to notice that the $\PPCFlc$ language and the
associated machines are only an \emph{intermediate step} in the proof
of Theorem~\ref{th:expect-time-diff}, the main result of this section,
whose statement does not mention them at all.

%

\subsection{Probabilistic PCF with labels and the associated random
  variables}\label{sec:ppcf-labels}
In order to count the number of times a given subterm $N$ of a closed
term $M$ of type $\Tnat$ is used (that is, arrives in head position)
during the execution of $\State M\Stempty$ in the Krivine machine of
Section~\ref{sec:PCF-operational}, we extend $\PPCF$ into $\PPCFlab$
by adding a term labeling construct $\Mark Nl$ for $l$ belonging to a
fixed set $\Labels$ of labels.
The typing rule for this new construct is simply
\begin{center}
  \AxiomC{$\Tseq\Gamma N\sigma$}
  \UnaryInfC{$\Tseq\Gamma{\Mark Nl}\sigma$}
  \DisplayProof
\end{center}
Of course $\PPCFlab$-stacks involve now such labeled terms but their
syntax is not extended otherwise; let $\Pcfstlab$ be the corresponding
set of states. Then we define a partial function
\[
  \Evalstlab:\Pcfstlab\times\Cantorfin\to\Mfin\Labels
\]
exactly as $\Evalst$ apart for the following cases,
\begin{linenomath}
\begin{align*}
  \Evalstlab(\State{\Mark Ml}{\pi},\alpha)
  &=\Evalstlab(\State{M}{\pi},\alpha)+\Mset l\\
   \Evalstlab(\State{\Dice r}{\pi},\Ocons i\alpha)
  &=\Evalstlab(\State{\Num i}{\pi},\alpha)\\
  \Evalstlab(\State{\Num 0}{\Stempty},\Oempty)
  &=0\quad\text{the empty multiset}.
\end{align*}
\end{linenomath}
When applied to $\State M\Stempty$, this function counts how often
labeled subterms of $M$ arrive in head position during the reduction;
these numbers, represented altogether as a multiset of elements of
$\Labels$, depend of course on the random tape provided as argument
together with the state. 

Let $\Evdomlab s$ be the set of $\alpha$'s such that
$\Evalstlab(s,\alpha)$ is defined.  Define $\Striplab s\in\Pcfst$ as
$s$ stripped from its labels.

\begin{lemma}\label{lemma:Evaldom-strip}
  For any $s\in\Pcfstlab$ we have
  $\Evdomlab{s}=\Evdom{\Striplab s}$.
\end{lemma}
\begin{proof}
  Simple inspection of the definition of the two involved
  functions. More precisely, an easy induction on $n$ shows that
  \begin{linenomath}
  \begin{align*}
    \forall n\in\Nat\quad
    \Evalst(\Striplab s,\alpha,n)\not=\Undef
    \Equiv
    \exists k\geq n\ \Evalstlab(s,\alpha,k)\not=\Undef\,.
  \end{align*}
  \end{linenomath}
\end{proof}

We define a
r.v.\footnote{That is, simply, a function since we are in a
  discrete probability setting.}
$\Evalstlabf s:\Cantorfin\to\Mfin\Labels$ by
\begin{linenomath}
\begin{equation*}
  \Evalstlabf s(\alpha)=
  \begin{cases}
    \Evalstlab(s,\beta) & \text{if }\alpha=\Ocons 0\beta
    \text{ and }\beta\in\Evdom{\Striplab s}\\
    0 & \text{in all other cases.}
  \end{cases}
\end{equation*}
\end{linenomath}
Let $l\in\Labels$ and let $\Evalstlabfp sl:\Cantorfin\to\Nat$ be the
\emph{integer} r.v.~defined by
$\Evalstlabfp sl(\alpha)=\Evalstlabf s(\alpha)(l)$. Its expectation is
\begin{linenomath}
\begin{align}\label{eq:esp-evalstlab-multiset}
  \begin{split}
  \Expect{\Evalstlabfp sl}
  &=\sum_{n\in\Nat}n\,\Proba s(\Evalstlabfp sl=n)\\
  &=\sum_{n\in\Nat}n\sum_{\Biind{\mu\in\Mfin\Labels}{\mu(l)=n}}
    \Proba s(\Evalstlabf s=\mu)\\
  &=\sum_{\mu\in\Mfin L}\mu(l)\Proba s(\Evalstlabf s=\mu)\,.
  \end{split}
\end{align}
\end{linenomath}
This is the expected number of occurrences of $l$-labeled
subterms of $s$ arriving in head position during successful executions
of $s$. It is more meaningful to condition this expectation under
convergence of the execution of $s$ (that is, under the event
$\Eventconvz{\Striplab s}$). We have
\begin{linenomath}
\begin{equation*}
  \Expect{\Evalstlabfp sl\mid\Eventconvz{\Striplab
      s}}=\frac{\Expect{\Evalstlabfp {s}l}}{\Proba {\Striplab
      s}(\Eventconvz{\Striplab s})}
\end{equation*}
\end{linenomath}
as the r.v.~$\Evalstlabfp sl$ vanishes outside the event
$\Eventconvz s$ since $\Evdomlab s=\Evdom{\Striplab s}$.



\subsection{A bird's eye view of the proof}
Our goal now is to extract this expectation from the denotational
semantics of a term $M$ such that $\Tseq{}M\Tnat$, which contains
labeled subterms, or rather of a term suitably definable from $M$.

For this purpose we will replace in $M$ each $\Mark Nl$ (where $N$ has
type $\sigma$) with $\If{x_l}{N}{\Loopt\sigma}$ where
$\Vect x=(x_l)_{l\in L}$ (for some finite subset $L$ of $\Labels$
containing all the labels occurring in $M$) is a family of pairwise
distinct variables of type $\Tnat$ and
$\Loopt\sigma=\Fix{\Abst x\sigma x}$ (an ever-looping term). We will
obtain in that way in Section~\ref{sec:spy-term} a term
$\Spy{\Vect x}M$ such that
\begin{linenomath}
\begin{align*}
  \Psem{\Spy{\Vect x}M}{\Vect x}\in\Kl\PCOH(\Snat^L,\Snat)\,.
\end{align*}
\end{linenomath}
We will consider this function as an analytic function
$(\Pcoh\Snat)^L\to\Pcoh\Snat$ which therefore induces an analytic
function
\begin{linenomath}
\begin{align*}
f:[0,1]^L&\to[0,1]\\
\Vect r&\mapsto\Psem{\Spy{\Vect x}M}{}((r_l\Base 0)_{l\in L})_0
\end{align*}
\end{linenomath}
(where $\Vect r\,\Base 0=(r_l\,\Base 0)_{l\in L}\in\Pcoh\Snat^L$ for
$\Vect r\in[0,1]^L$). We will prove that the expectation of the
number of uses of subterms of $M$ labeled by $l$ is
\begin{linenomath}
  \begin{align*}
  \frac{\partial f(\Vect r)}{\partial r_l}(1,\dots,1)\,.
  \end{align*}
\end{linenomath}
Notice that in the partial derivative above, the $r_l$'s are bound by
the partial derivative itself and by the fact that it is evaluated at
$(1,\dots,1)$.

In order to reduce this problem to Theorem~\ref{th:pcoh-adequacy}, we
will introduce the intermediate language $\PPCFlc$ which will allow to see
each of these parameters $r_l$ as the probability of yielding $\Num 0$
for a biased coin construct $\Dicelab l{r_l}$. This calculus will be
executed in a further ``Krivine machine'' $\Evalstlc$ which has as
many random tapes as there are elements in $L$ (plus one for the plain
$\Dice\_$ constructs already occurring in $M$).

This intermediate language will be used as follows: given a closed
labeled term $M$ whose all labels belong to $L\subseteq\Labels$ and a
family of probabilities $\Vect r\in\Intercc 01^L$ we will define in
Section~\ref{sec:spy-term} a term $\Lcof{\Vec r}(M)$ of $\PPCFlc$
which is $M$ where each labeled subterm $\Mark Nl$ is replaced with
$\If{\Dicelab l{r_l}} {\Lcof{\Vect r}(N)}{\Loopt\sigma}$ where
$\Loopt\sigma$. The term $\Lcof{\Vect r}(M)$ defined in that way is
{closed}, the $r_l$'s being {parameters} and not {variables} in the
sense of the $\lambda$-calculus. The probability $p(\Vect r)$ of
convergence of $\Lcof{\Vect r}(M)$ depends on $(\mu(l))_{l\in I}$
where $\mu(l)\in\Nat$ is the number of times an $l$-labeled subterm of
$M$ arrives in head-position during the evaluation of $M$: this is the
meaning of Lemma~\ref{lemma:Lcof_evaldom_struct-3}. More precisely
$\mu(l)$ is the exponent of $r_l$'s in this probability. The main
feature of $\Spy{\Vect x}(M)$, exploited in
Section~\ref{sec:sp-trans}, is that $p(\Vect r)$ is obtained by
applying the semantics of $\Spy{\Vect x}(M)$ to $\Vect r$ (or more
precisely to $(r_l\Base 0)_{l\in L}\in\Pcoh\Snat^L$) ---~the proof of
this fact uses Theorem~\ref{th:pcoh-adequacy}. The last step will
consist in observing that, by taking the derivative of this
probability wrt.~the variables $x_l$, one obtains the expectation of
the number of times an $l$-labeled term arrives in head-position
during the evaluation of $M$; this is due to the fact that the
$\mu(l)$ are exponents in the expression of $p(\Vect r)$ and that these
exponents become coefficients by differentiation.

\subsection{Probabilistic PCF with labeled coins}\label{sec:PCF-lc}
Let $\PPCFlc$ be $\PPCF$ extended with a 
construct $\Dicelab lr$ typed as
\begin{center}

  \AxiomC{$r\in[0,1]\cap\Rational$ and $l\in\Labels$}
  \UnaryInfC{$\Tseq\Gamma{\Dicelab lr}\Tnat$}
  \DisplayProof
\end{center}
This language features the usual $\Dice r$ construct for probabilistic
choice as well as a supply of identical constructs labeled by
$\Labels$ that we will use to simulate the counting of
Section~\ref{sec:ppcf-labels}. Of course $\PPCFlc$-stacks involve now
terms with labeled coins but their syntax is not extended otherwise;
let $\Pcfstlc$ be the corresponding set of states. We use $\Labof M$
for the set of labels occurring in $M$ (and similarly $\Labof s$ for
$s\in\Pcfstlc$). Given a \emph{finite} subset $L$ of $\Labels$, we use
$\PPCFlc(L)$ for the set of terms $M$ such that $\Labof M\subseteq L$
and we define similarly $\Pcfstlc(L)$. We use the similar notations
$\PPCFlab(L)$ and $\Pcfstlab(L)$ for the sets of labeled terms and
stacks (see Section~\ref{sec:ppcf-labels}) whose all labels belong to
$L$.

The partial function
$\Evalstlc:\Pcfstlc(L)\times\Cantorfin\times\Cantorfin^L\to\Realp$ is
defined exactly as $\Evalst$ (for the unlabeled $\Dice r$, we use only
the first parameter in $\Cantorfin$), with the additional parameters
$\Vect\beta$ passed unchanged in the recursive calls, apart for the
the following new rules:
\begin{linenomath}
  \begin{align*}
  \Evalstlc(\State{\Dicelab lr}{\pi},\alpha,\Vect\beta)
  = \Evalstlc(\State{\Num i}{\pi},\alpha,\Subst{\Vect\beta}\gamma
      l)\cdot \Pneg ir\text{\quad if }\beta(l)=\Ocons i\gamma
  \end{align*}
\end{linenomath}
where $\Vect \beta=(\beta(l))_{l\in L}$ stands for an $L$-indexed
family of elements of $\Cantorfin$ and $\Subst{\Vect\beta}\gamma l$ is
the family $\Vect\delta$ such that $\delta(l')=\beta(l')$ if $l'\not=l$
and $\delta(l)=\gamma$. We define
$\Evdomlc s\subseteq\Cantorfin\times\Cantorfin^L$ as the domain of the
partial function $\Evalstlc(s,\_,\_)$.

  We define a version $\Evalstlcsh(s,\_,\_)$ of the machine
  $\Evalstlc(s,\_,\_)$ which returns an element of $\Cantorfin$
  instead of an element of $\Realp$. The definition is the same up to
  the following rules:
  \begin{linenomath}
  \begin{align*}
    \Evalstlcsh(\State{\Dicelab lr}{\pi},\alpha,\Vect\beta)
    &= \Ocons i{\Evalstlcsh(\State{\Num i}{\pi},\alpha,\Subst{\Vect\beta}\gamma
      l)}\text{\quad if }\beta(l)=\Ocons i\gamma\\
    \Evalstlcsh(\State{\Dice r}{\pi},\Ocons i\alpha,\Vect\beta)
    &= \Ocons i{\Evalstlcsh(\State{\Num i}{\pi},\alpha,\Vect\beta)}\\
    \Evalstlcsh(\State{\Num 0}{\Stempty},\Oempty,\Vect\Oempty)
    &= \Oempty\,.
  \end{align*}
  \end{linenomath}
  When defined,
  $\Evalstlcsh(\State{\Dicelab lr}{\pi},\alpha,\Vect\beta)$ is a
  shuffle of the $\beta_l$'s and of $\alpha$ (in the order in which
  the corresponding elements of the tape are read during the
  execution).  Being defined by similar recursions, the functions
  $\Evalstlc(s,\_,\_)$ and $\Evalstlcsh(s,\_,\_)$ have the same
  domain. Then we have
  \begin{linenomath}
  \begin{align}\label{eq:evalstlc-evalstlcsh-comp}
    \Evalstlc(s,\alpha,\Vect\beta)
    =\Evalst(\Striplc s,\Evalstlcsh(s,\alpha,\Vect\beta))
  \end{align}
  \end{linenomath}
  where ``$=$'' should be understood as Kleene equality (either both
  sides are undefined or both sides are defined and equal). 

  To prove Equation~\Eqref{eq:evalstlc-evalstlcsh-comp}, one considers
  step-indexed versions of the involved machines, equipped with a
  further parameter in $\Nat$. For instance the modified
  $\Evalst(s,\alpha,n)$ will be a total function
  $\Pcfstlc(L)\times(\Cantorfin\cup\Eset\Undef)\times\Nat
  \to\Realp\cup\Eset\Undef$ where $\Undef$ stands for non-terminated
  computations. Here are a few cases of this modified definition, the
  others being similar:
  \begin{linenomath}
  \begin{align}\label{eq:evalst-indexed-def}
    \begin{split}
      \Evalst(s,\alpha,0)
      &= \Undef\\
      \Evalst(\State{\Let xMN}\pi,\alpha,n+1)
      &= \Evalst(\State M{\Stlet xN\Stcons\pi},\alpha,n)\\
      \Evalst(\State{\Dice r}{\pi},\Ocons 0\alpha,n+1)
      &= \Evalst(\State{\Num i}{\pi},\alpha,n)\cdot \Pneg ir\\
      \Evalst(\State{\Dice r}{\pi},\Oempty,n+1)
      &= \Undef\\
      \Evalst(\State{\Dice r}{\pi},\Undef,n+1)
      &= \Undef\\
      \Evalst(\State{\Num 0}{\Stempty},\Oempty,n+1)&= 1
    \end{split}
  \end{align}
  \end{linenomath}
  where of course multiplication is extended by
  $r\Undef=\Undef r=\Undef$ and similarly for concatenation.
  \begin{remark}
    One obvious and important feature of this definition, shared by
    all the forthcoming definitions based on similar step-indexing, is
    that if $\Evalst(s,\alpha,n)\not=\Undef$, we have
    $\Evalst(s,\alpha,k)=\Evalst(s,\alpha,n)\not=\Undef$ for all
    $k\geq n$. This property will be referred to as ``monotonicity of
    step-indexing''.
  \end{remark}
  Then $\Evalst(s,\alpha)$ is defined, and has value $r$, iff
  $\Evalst(s,\alpha,n)=r$ for some $n$ (and then the same will hold
  for any greater $n$).

  Thanks to this step-indexing, the proof of
  Equation~\Eqref{eq:evalstlc-evalstlcsh-comp} boils down to the
  following lemma.
  \begin{lemma}\label{lemma:evalst-evalstlcsh}
  For all $n\in\Nat$, one has
  \begin{linenomath}
    \begin{align*}
      \forall n\in\Nat\quad
      \Evalstlc(s,\alpha,\Vect\beta,n)
      =\Evalst(\Striplc s,\Evalstlcsh(s,\alpha,\Vect\beta,n),n)\,.
    \end{align*}
  \end{linenomath}    
  \end{lemma}
  \begin{proof}

  Assume that the property holds for all integer $p<n$ and let us
  prove it for $n$. One reasons by cases on the shape of $s$
  considering only a few cases, the others being similar. Notice that
  the equation is obvious if $n=0$ since then both hand-sides are
  $=\Undef$ so we can assume $n>0$.
  \begin{itemize}
  \item $s=\State{\Let xMN}\pi$. By definition of the machine
    $\Evalstlc$ we have
    $\Evalstlc(s,\alpha,\Vect\beta,n)=\Evalstlc(t,\alpha,\Vect\beta,n-1)$
    and
    $\Evalst(\Striplc s,\Evalstlcsh(s,\alpha,\Vect\beta,n),n)
    =\Evalst(\Striplc t,\Evalstlcsh(t,\alpha,\Vect\beta,n-1),n-1)$
    where $t=\State M{\Stlet xN\Stcons\pi}$ and the inductive
    hypothesis applies.
  \item $s=\State{\Dice r}{\pi}$. We have
    $\Evalstlc(s,\alpha,\Vect\beta,n)=\Undef=\Evalst(\Striplc
    s,\Evalstlcsh(s,\alpha,\Vect\beta,n),n)$ if $\alpha=\Oempty$ or
    $\alpha=\Undef$. We have
    $\Evalstlc(s,\Ocons i\alpha,\Vect\beta,n)
    =\Evalstlc(t,\alpha,\Vect\beta,n-1)\cdot\Pneg ir$ and, setting
    $t=\State{\Num i}\pi$,
    \begin{linenomath}
    \begin{align*}
      \Evalst(\Striplc s,\Evalstlcsh(s,\Ocons i\alpha,\Vect\beta,n),n)
      &=\Evalst(\Striplc s,\Ocons i{\Evalstlcsh(t,\alpha,\Vect\beta,n-1)},n)\\
      & =\Evalst(\Striplc t,{\Evalstlcsh(t,\alpha,\Vect\beta,n-1)},n-1)
        \cdot\Pneg ir 
    \end{align*}
    \end{linenomath}
    and the inductive hypothesis applies.
  \item As a last example assume that $s=\State{\Dicelab lr}{\pi}$. We
    have
    $\Evalstlc(s,\alpha,\Vect\beta,n)=\Undef=\Evalst(\Striplc
    s,\Evalstlcsh(s,\alpha,\Vect\beta,n),n)$ if $\beta(l)=\Oempty$ or
    $\beta(l)=\Undef$. If $\beta(l)=\Ocons i\gamma$ we have
    $\Evalstlc(s,\alpha,\Vect\beta,n) =\Evalstlc(\State{\Num
      i}{\pi},\alpha,\Subst{\Vect\beta}\gamma l,n-1)\cdot\Pneg ir$ and
    \begin{linenomath}
    \begin{align*}
      \Evalst(\Striplc s,\Evalstlcsh(s,\alpha,\Vect\beta,n),n)
      &=\Evalst(\State{\Dice r}{\Striplc\pi}, \Ocons
        i{\Evalstlcsh(\State{\Num i}{\pi},\alpha, \Subst{\Vect\beta}\gamma
        l,n-1)},n)\\
      &=\Evalst(\State{\Num i}{\Striplc\pi},
        {\Evalstlcsh(\State{\Num i}{\pi},\alpha, \Subst{\Vect\beta}\gamma
        l,n-1)},n-1)\cdot\Pneg ir
    \end{align*}
    \end{linenomath}
    and the inductive hypothesis applies.
  \end{itemize}
  \end{proof}

Let $\Striplc s\in\Pcfst$ be obtained by stripping $s$ from its labels
(so that $\Striplc{\Dicelab lr}=\Dice r$).  And $\Striplc M\in\PPCF$
is defined similarly.

  Equation~\Eqref{eq:evalstlc-evalstlcsh-comp} shows in particular that
  \begin{linenomath}
    \begin{align*}
      \forall(\alpha,\Vect\beta)\in\Evdomlc s\quad
      \Evalstlcsh(s,\alpha,\Vect\beta)\in\Evdom{\Striplc s}\,.
    \end{align*}
  \end{linenomath}

  \begin{lemma}\label{lemma:Evalstlcsh-inv}
  The function
  \begin{linenomath}
  \begin{align*}
    \eta_s:\Evdomlc s & \to \Evdom{\Striplc s}=\Evdomlab s\\
    (\alpha,\Vect\beta) & \mapsto \Evalstlcsh(s,\alpha,\Vect\beta)
  \end{align*}
  \end{linenomath}
  is a bijection.
  \end{lemma}

  \begin{proof}
    We provide explicitly an
  inverse function, defined as another machine
  $\Evalstlcshinv(s,\alpha)$. Again we provide only a few cases
  \begin{linenomath}
  \begin{align*}
    \Evalstlcshinv(\State{\Let xMN}{\pi},\delta)
    &=\Evalstlcshinv(\State{M}{\Stlet xN\Stcons\pi},\delta)\\
    \Evalstlcshinv(\State{\Dicelab lr}{\pi},\Ocons i\delta)
    &= (\alpha,\Vect\gamma)
      \text{\quad if }\Evalstlcshinv(\State{\Num i}{\pi},\delta)
      =(\alpha,\Vect\beta),\\
      &\hspace{2cm}\gamma(l)=\Ocons i{\beta(l)}
      \text{ and }\gamma(k)=\beta(k)\text{ for }k\not=l\\
    \Evalstlcshinv(\State{\Dice r}{\pi},\Ocons i\delta)
    &= (\Ocons i\alpha,\Vect\beta)
      \text{\quad if }\Evalstlcshinv(\State{\Num i}{\pi},\delta)
      =(\alpha,\Vect\beta)\\
    \Evalstlcshinv(\State{\Num 0}{\Stempty},\Oempty)
    &= (\Oempty,\Vect\Oempty)\,.
  \end{align*}
  \end{linenomath}
  It is clear from this recursion that the partial function
  $\Evalstlcshinv(s,\_)$ has $\Evdom{\Striplc s}$ as domain.
  Let us prove that
  \begin{linenomath}
  \begin{align*}
    \forall(\alpha,\Vect\beta)\in\Evdomlc s
    \quad \Evalstlcshinv(s,\Evalstlcsh(s,\alpha,\Vect\beta))
    =(\alpha,\Vect\beta)\,.
  \end{align*}
  \end{linenomath}
  Considering a step-indexed version of $\Evalstlcshinv$ with an
  additional parameter in $n\in\Nat$ defined along the same lines
  as~\Eqref{eq:evalst-indexed-def}, it suffices to prove that
  \begin{linenomath}
  \begin{align}\label{eq:evalstlcshinv-evalstlcsh-step}
    \forall n\in\Nat\,\forall(\alpha,\Vect\beta)\in\Evdomlc s
    \quad
    \Evalstlcsh(s,\alpha,\Vect\beta,n)\not=\Undef
    \Implies
    \Evalstlcshinv(s,\Evalstlcsh(s,\alpha,\Vect\beta,n),n)
    =(\alpha,\Vect\beta)\,.
  \end{align}
  \end{linenomath}
  The proof is by induction on $n$.
  So assume that the property holds for all integer $p<n$ and let us
  prove it for $n$. Assume that
  $\Evalstlcsh(s,\alpha,\Vect\beta,n)\not=\Undef$, which implies that
  $n>0$.  As usual we consider only a few interesting cases. All other
  cases are similar to the first one (and similarly trivial).
  \begin{itemize}
  \item Assume first that $s=\State{\Let xMN}\pi$. Setting
    $t=\State M{\Stlet xN\Stcons\pi}$ we have
    $\Evalstlcsh(t,\alpha,\Vect\beta,n-1)
    =\Evalstlcsh(s,\alpha,\Vect\beta,n)\not=\Undef$. Then
    \begin{linenomath}
    \begin{align*}
      \Evalstlcshinv(s,\Evalstlcsh(s,\alpha,\Vect\beta,n),n)
      &=\Evalstlcshinv(t,\Evalstlcsh(t,\alpha,\Vect\beta,n-1),n-1)\\
      &=(\alpha,\Vect\beta)
    \end{align*}
    \end{linenomath}
    by inductive hypothesis.
  \item Assume now that $s=\State{\Dice r}{\pi}$. Since
    $\Evalstlcsh(s,\alpha,\Vect\beta,n)\not=\Undef$ we must have
    $\alpha\not=\Oempty$. Let us write $\alpha=\Ocons i\gamma$, then
    we have
    $\Evalstlcsh(s,\alpha,\Vect\beta,n)=\Ocons
    i{\Evalstlcsh(t,\gamma,\Vect\beta,n-1)}$ where
    $t=\State{\Num i}\pi$ and we have
    $\Evalstlcsh(t,\gamma,\Vect\beta,n-1)\not=\Undef$. By inductive
    hypothesis it follows that
    \begin{linenomath}
    \begin{align*}
      \Evalstlcshinv(t,\Evalstlcsh(t,\gamma,\Vect\beta,n),n)
      =(\gamma,\Vect\beta)
    \end{align*}
    \end{linenomath}
    Then we have
    \begin{linenomath}
    \begin{align*}
      \Evalstlcshinv(s,\Evalstlcsh(s,\alpha,\Vect\beta,n),n)
      &=\Evalstlcshinv(s,\Ocons i{\Evalstlcsh(t,\gamma,\Vect\beta,n-1)},n)\\
      &=(\Ocons i\gamma,\Vect\beta)\\
      &=(\alpha,\Vect\beta)
    \end{align*}
    \end{linenomath}
    by definition of $\Evalstlcshinv$.
  \item The case $s=\State{\Dicelab lr}{\pi}$ is similar to the
    previous one, dealing with $\beta(l)$ instead of $\alpha$.
  \end{itemize}

  Now we prove that for all $\delta\in\Evdom{\Striplc s}$ one has
  $\Evalstlcshinv(s,\delta)\in\Evdomlc s$ and that
  \begin{linenomath}
  \begin{align*}
    \forall\delta\in\Evdom{\Striplc s}
    \quad \Evalstlcsh(s,\Evalstlcshinv(s,\delta))=\delta\,.
  \end{align*}
  \end{linenomath}
  It suffices to prove that
  \begin{linenomath}
  \begin{align*}
    \forall n\in\Nat\,\forall\alpha\in\Evdom{\Striplc s}
    \quad \Evalstlcshinv(s,\delta,n)\not=\Undef
    \Implies \Evalstlcsh(s,\Evalstlcshinv(s,\delta,n),n)=\delta
  \end{align*}
  \end{linenomath}
  using as usual the step-indexed versions of our machines. The proof
  is by induction on $n$ so assume that the property holds for all
  $p<n$ and let us prove it for $n$. Assume that
  $\Evalstlcshinv(s,\delta,n)\not=\Undef$, which implies $n>0$. We
  reason by cases on $s$, considering the same cases as above (the
  other cases, similar to the first one, are similarly trivial).
  \begin{itemize}
  \item Assume first that $s=\State{\Let xMN}\pi$. Setting
    $t=\State M{\Stlet xN\Stcons\pi}$ we know that
    $\Evalstlcshinv(t,\delta,n-1)\not=\Undef$ and hence by inductive
    hypothesis $\Evalstlcsh(t,\Evalstlcshinv(t,\delta,n-1),n-1)=\delta$,
    proving our contention.
  \item Assume next that $s=\State{\Dice r}{\pi}$. Since
    $\Evalstlcshinv(s,\delta,n)\not=\Undef$ we know that
    $\delta\not=\Oempty$. So we can write $\delta=\Ocons
    i\theta$. Then we know that
    $\Evalstlcshinv(s,\delta,n)=(\Ocons i\alpha,\Vect\beta)$ where
    $(\alpha,\Vect\beta)=\Evalstlcshinv(\State{\Num
      i}\pi,\delta,n-1)$. By inductive hypothesis we have
    $\Evalstlcsh(\State{\Num i}\pi,(\alpha,\Vect\beta))=\theta$ and therefore
    \begin{linenomath}
    \begin{align*}
      \Evalstlcsh(s,\Evalstlcshinv(s,\delta,n),n)
      &=\Evalstlcsh(s,(\Ocons i\alpha,\Vect\beta),n)\\
      &=\Ocons i\theta\\
      &=\delta\,.
    \end{align*}
    \end{linenomath}
  \item The case $s=\State{\Dicelab lr}{\pi}$ is similar to the
    previous one, dealing with $\beta(l)$ instead of $\alpha$.
  \end{itemize}
\end{proof}

\begin{lemma}\label{lemma:proba-conv-lc}
  For all $s\in\Pcfstlc(L)$
  \begin{linenomath}
    \begin{align*}
  \Proba{\Striplab s}(\Eventconvz{\Striplab s})
  =\sum_{(\alpha,\Vect\beta)\in\Evdomlc s}
  \Evalstlc(s,{\alpha,\Vect\beta})\,.
    \end{align*}
    \end{linenomath}
\end{lemma}
\begin{proof}

  We have
  \begin{linenomath}
  \begin{align*}
    \Proba{\Striplab s}(\Eventconvz{\Striplab s})
    &=\sum_{\delta\in\Evdom{\Striplc s}}\Evalst(\Striplc s,\delta)\\
    &=\sum_{(\alpha,\Vect\beta)\in\Evdomlc s}
      \Evalstlc(s,{\alpha,\Vect\beta})\,.
  \end{align*}
  \end{linenomath}
  By equation~\Eqref{eq:evalstlc-evalstlcsh-comp}
  and by the bijective correspondence of Lemma~\ref{lemma:Evalstlcsh-inv}.
\end{proof}

\subsection{Spying labeled terms in $\PPCF$}\label{sec:spy-term}
Given $\Vect r=(r_l)_{l\in L}\in(\Rational\cap[0,1])^L$, we define a
(type preserving) translation $\Lcof{\Vect r}:\PPCFlab(L)\to\PPCFlc$
by induction on terms. For all term constructs but labeled terms, the
transformation does nothing (for instance $\Lcof{\Vect r}(x)=x$,
$\Lcof{\Vect r}(\Abst x\sigma M)=\Abst x\sigma{\Lcof{\Vect r}(M)}$
etc), the only non-trivial case being
\begin{linenomath}
\begin{align*}
  \Lcof{\Vect r}(\Mark Ml)=\If{\Dicelab l{r_l}}
  {\Lcof{\Vect r}(M)}{\Loopt\sigma}
\end{align*}
\end{linenomath}
where $\sigma$ is the type\footnote{\emph{A priori} this type is known
  only if we know the type of the free variables of $M$, so to be more
  precise this translation should be specified in a given typing
  context; this can easily be fixed by adding a further parameter to
  $\Lcof{}$ at the price of heavier notations.} of $M$. 

In Section~\ref{sec:sp-trans}, we will turn a \emph{closed} labeled
term $M$ (with labels in the finite set $L$) into the term
$\Spy{\Vect x}(M)$, defined in such a way that
$\Psem{\Striplc{\Lcof{\Vect r}(M)}}{}$ has a simple expression in
terms of $\Spy{\Vect x}(M)$ (Lemma~\ref{lemma:spy-lc-expression}),
allowing to interpret the coefficients of the power series
interpreting $\Spy{\Vect x}(M)$ in terms of probability of reduction
of the machine $\Evalstlab$ with given resulting multisets of labels
(Equation~\Eqref{eq:Evsp-coeff-proba}). This in turn is the key to the
proof of Theorem~\ref{th:expect-time-diff}.

We write $\Osingle 0^k$ for
the sequence consisting of $k$ occurrences of $0$.

\begin{lemma}\label{lemma:Lcof_evaldom_struct-2}
  Let $s\in\Pcfstlab(L)$. Then
  \begin{align*}
    \Evdomlc{\Lcof{\Vect r}(s)}
    =\{(\alpha,(\Osingle
      0^{\Evalstlab(s,\alpha)(l)})_{l\in L})\St\alpha\in\Evdom{\Striplab
      s}\}
  \end{align*}
\end{lemma}
Remember from Lemma~\ref{lemma:Evaldom-strip} that
$\Evdomlab{s}=\Evdom{\Striplab s}$.
\begin{proof}
  We show that for any $n\in\Nat$ and any
  $(\alpha,\Vect\beta)\in\Cantorfin\times\Cantorfin^L$, one has
  \begin{linenomath}
  \begin{align}
    \begin{split}\label{eq:Lcof_evaldom_struct-2-step-1}
    &\forall n\in\Nat\
    \big(\Evalstlc(\Lcof{\Vec r}(s),\alpha,\Vect\beta,n)\not=\Undef\\
    &\hspace{8em}\Implies
    \exists k\in\Nat\ 
    \Evalstlab(s,\alpha,k)\not=\Undef
    \text{ and }
    \forall l\in L\ \beta_l=\Osingle 0^{\Evalstlab(s,\alpha,k)(l)}\big)\\
    \end{split}\\
    \begin{split}\label{eq:Lcof_evaldom_struct-2-step-2}
    &\forall n\in\Nat\ 
    \big(\Evalstlab(s,\alpha,n)\not=\Undef
    \text{ and }
      \forall l\in L\ \beta_l=\Osingle 0^{\Evalstlab(s,\alpha,n)(l)}\\
    &\hspace{8em}\Implies
    \exists k\in\Nat\
    \Evalstlc(\Lcof{\Vec r}(s),\alpha,\Vect\beta,k)\not=\Undef\big)
    \end{split}
  \end{align}
  \end{linenomath}
  using as before
  step-indexed versions of the various machines. But in the present
  situation we shall not have the same indexing on both sides of
  implications because the encoding $\Lcof{\Vect r}(s)$ requires
  additional execution steps.

  We prove~\Eqref{eq:Lcof_evaldom_struct-2-step-1} by induction on
  $n\in\Nat$.
  Assume now that the implication holds for all integers $p<n$ and let
  us prove it for $n$, so assume that
  $\Evalstlc(\Lcof{\Vec r}(s),\alpha,\Vect\beta,n)\not=\Undef$ which
  implies $n>0$. We consider only three cases as to the shape of $s$,
  the other cases being completely similar to the first one. We use
  the following convention: if $M$ is a labeled term we use
  $M'$ for $\Lcof{\Vect r}(M)$ and similarly for stacks and states.
  \begin{itemize}
  \item Assume first that $s=\State{\Let xMN}\pi$ and let
    $t=\State M{\Stlet xN\Stcons\pi}$. We have
    $s'=\State{\Let x{M'}{N'}}{\pi'}$ and
    $t'=\State {M'}{\Stlet x{N'}\Stcons{\pi'}}$.  We have
    \begin{align*}
      \Evalstlc(t',\alpha,\Vect\beta,n-1)
      =\Evalstlc(s',\alpha,\Vect\beta,n)\not=\Undef
    \end{align*}
    and hence by inductive hypothesis there is $k\in\Nat$ such that
    \begin{linenomath}
      \begin{align*}
      \Evalstlab(t,\alpha,k)\not=\Undef
      \text{ and }
      \forall l\in L\ \beta_l=\Osingle 0^{\Evalstlab(t,\alpha,k)(l)}\,.
      \end{align*}
    \end{linenomath}
    It follows that
    $\Evalstlab(s,\alpha,k+1)
    =\Evalstlab(t,\alpha,k)\not=\Undef$ and
    $\Evalstlab(s,\alpha,k+1)(l)=\Evalstlab(t,\alpha,k)(l)$. Therefore
    we have
    \begin{linenomath}
    \begin{align*}
      \Evalstlab(s,\alpha,k+1)\not=\Undef
      \text{ and }
      \forall l\in L\ \beta_l=\Osingle 0^{\Evalstlab(s,\alpha,k+1)(l)}\,.
    \end{align*}
    \end{linenomath}
  \item Assume now that $s=\State{\Dice r}{\pi}$ so that
    $s'=\State{\Dice r}{\pi'}$. Since
    $\Evalstlc(s',\alpha,\Vect\beta,n)\not=\Undef$ we know that
    $\alpha=\Ocons i\gamma$ for some $i\in\Eset{0,1}$ and that
    $\Evalstlc(t',\gamma,\Vect\beta,n-1)\not=\Undef$ where
    $t=\State{\Num i}\pi$ (and hence $t'=\State{\Num i}{\pi'}$). By
    inductive hypothesis there is $k\in\Nat$ such that
    \begin{linenomath}
    \begin{align*}
      \Evalstlab(t,\gamma,k)\not=\Undef
      \text{ and }
      \forall l\in L\ \beta_l=\Osingle 0^{\Evalstlab(t,\gamma,k)(l)}\,.
    \end{align*}
    \end{linenomath}
    It follows that
    $\Evalstlab(s,\alpha,k+1)
    =\Evalstlab(t,\gamma,k)\not=\Undef$. Therefore
    we have
    \begin{linenomath}
    \begin{align*}
      \Evalstlab(s,\alpha,k+1)\not=\Undef
      \text{ and }
      \forall l\in L\ \beta_l=\Osingle 0^{\Evalstlab(s,\alpha,k+1)(l)}\,.
    \end{align*}
    \end{linenomath}
  \item The third case we consider is $s=\State{\Mark Ml}\pi$ for some
    $l\in L$ so that
    \begin{align*}
      s'=\State{\If{\Dicelab l{r_l}}{M'}{\Loopt\sigma}}{\pi'}
    \end{align*}
    where $\sigma$ is the type of $M$. Since
    $\Evalstlc(s',\alpha,\Vect\beta,n)\not=\Undef$ we must have
    $n\geq 2$ and $\beta_l=\Ocons i\gamma$; indeed, setting
    $\Vect{\beta'}=\Subst{\Vect\beta}{\gamma}{l}$,
    \begin{linenomath}
    \begin{align*}
      \Evalstlc(s',\alpha,\Vect\beta,n)
      &= \Evalstlc(\State{\Dicelab l{r_l}}
        {\Stif{M'}{\Loopt\sigma}\Stcons\pi'},\alpha,\Vect\beta,n-1)\\
      &= \Evalstlc(\State{\Num i}{\Stif{M'}{\Loopt\sigma}\Stcons\pi'},
          \alpha,\Vect{\beta'},n-2)\\
      &=
        \begin{cases}
          \Evalstlc(\State{M'}{\pi'},\alpha,\Vect{\beta'},n-2)
          &\text{if }i=0\\
          \Evalstlc(\State{\Loopt\sigma}{\pi'},\alpha,\Vect{\beta'},n-2)
          &\text{if }i=1\,.
        \end{cases}
    \end{align*}
    \end{linenomath}
    But whatever is the value of $n\geq 2$ we have
    $\Evalstlc(\State{\Loopt\sigma}{\pi'},\alpha,\Vect{\beta'},n-2)=\Undef$
    by definition of $\Loopt\sigma$. It follows that we must have
    $i=0$ and
    $\Evalstlc(\State{M'}{\pi'},\alpha,\Vect{\beta'},n-2)\not=\Undef$.
    By inductive hypothesis there exists $k\in\Nat$ such that, setting
    $t=\State M\pi$
    \begin{linenomath}
    \begin{align*}
      \Evalstlab(t,\alpha,k)\not=\Undef
      \text{ and }
      \forall m\in L\ \beta'_m=\Osingle 0^{\Evalstlab(t,\alpha,k)(m)}\,.
    \end{align*}
    \end{linenomath}
    We have
    $\Evalstlab(s,\alpha,k+1)=\Evalstlab(s,\alpha,k)+\Mset
    l\not=\Undef$. It follows that if $m\in L\setminus\Eset l$ one has
    $\beta_m=\beta'_m=\Osingle 0^{\Evalstlab(t,\alpha,k)(m)}=\Osingle
    0^{\Evalstlab(s,\alpha,k+1)(m)}$, and
    $\beta_l=\Ocons 0{\beta'_l}=\Ocons 0{\Osingle
      0^{\Evalstlab(t,\alpha,k)(l)}}=\Ocons
    0^{\Evalstlab(s,\alpha,k+1)(l)}$ proving our contention.
  \end{itemize}
  This ends the proof of~\Eqref{eq:Lcof_evaldom_struct-2-step-1}, we
  prove now~\Eqref{eq:Lcof_evaldom_struct-2-step-2}, also by induction
  on $n$. Assume that the implication holds for all $p<n$ and let us
  prove it for $n$ so assume that
  \begin{linenomath}
    \begin{align*}
      \Evalstlab(s,\alpha,n)\not=\Undef
      \text{ and }
      \forall l\in L\ \beta_l=\Osingle 0^{\Evalstlab(s,\alpha,n)(l)}\,.
    \end{align*}
  \end{linenomath}
  As usual this implies that $n>0$.  We deal with the same three cases
  as in the proof of~\Eqref{eq:Lcof_evaldom_struct-2-step-1}.
  \begin{itemize}
  \item Assume first that $s=\State{\Let xMN}\pi$ and let
    $t=\State M{\Stlet xN\Stcons\pi}$. We have
    $s'=\State{\Let x{M'}{N'}}{\pi'}$ and
    $t'=\State {M'}{\Stlet x{N'}\Stcons{\pi'}}$. We have
    $\Evalstlab(t,\alpha,n-1)=\Evalstlab(s,\alpha,n)\not=\Undef$. Also we
    have, for all $l\in L$,
    $\beta_l=\Osingle 0^{\Evalstlab(s,\alpha,n)(l)}=\Osingle
    0^{\Evalstlab(t,\alpha,n-1)(l)}$. Hence by ind.~hypothesis
    there is $k\in\Nat$ such that
    $\Evalstlc(t',\alpha,\Vect\beta,k)\not=\Undef$ and hence
    $\Evalstlc(s',\alpha,\Vect\beta,k+1)
    =\Evalstlc(t',\alpha,\Vect\beta,k)\not=\Undef$.
  \item Assume now that $s=\State{\Dice r}{\pi}$ so that
    $s'=\State{\Dice r}{\pi'}$. Since
    $\Evalstlab(s,\alpha,n)\not=\Undef$ we must have 
    $\alpha=\Ocons i\gamma$ for some $i\in\Eset{0,1}$ and we have
    $\Evalstlab(t,\gamma,n-1)=\Evalstlab(s,\alpha,n)\not=\Undef$ where
    $t=\State{\Num i}{\pi}$. Moreover
    $\Evalstlab(s,\alpha,n)=\Evalstlab(t,\gamma,n-1)$ and hence for
    each $l\in L$ we have
    $\beta_l=\Osingle 0^{\Evalstlab(s,\alpha,n)(l)}=\Osingle
    0^{\Evalstlab(t,\gamma,n-1)(l)}$. By inductive hypothesis there
    exists $k$ such that
    $\Evalstlc(t',\gamma,\Vect\beta,k)\not=\Undef$. We have
    $\Evalstlc(s',\alpha,\Vect\beta,k+1)
    =\Evalstlc(t',\gamma,\Vect\beta,k)\not=\Undef$ as expected.
  \item Assume last that $s=\State{\Mark Ml}\pi$ for some $l\in L$ so
    that $s'=\State{\If{\Dicelab l{r_l}}{M'}{\Loopt\sigma}}{\pi'}$
    where $\sigma$ is the type of $M$.  Let $t=\State M\pi$.  Since
    $\Evalstlab(s,\alpha,n)\not=\Undef$ we have
    $\Evalstlab(t,\alpha,n-1)\not=\Undef$ and
    $\Evalstlab(s,\alpha,n)=\Evalstlab(t,\alpha,n-1)+\Mset l$.
    We also know that for all $m\in L$
    \begin{linenomath}
    \begin{align*}
      \beta_l=\Osingle 0^{\Evalstlab(s,\alpha,n)(m)}
    \end{align*}
    \end{linenomath}
    In particular $\beta_l=\Ocons 0{\beta'_l}$. Setting
    $\beta'_m=\beta_m$ for $m\not=l$, we have therefore
    \begin{linenomath}
    \begin{align*}
      \forall m\in L\quad \beta'_m=\Osingle 0^{\Evalstlab(t,\alpha,n)(m)}\,.
    \end{align*}
    \end{linenomath}
    By inductive
    hypothesis there is $k\in\Nat$ such that
    $\Evalstlc(\Lcof{\Vect
      r}(t),\alpha,\Vect{\beta'},k)\not=\Undef$. Since
    $s'=\Lcof{\Vect r}(s)=\If{\Dicelab l{r_l}}{M'}{\Loopt\sigma}$,
    setting $\Vect\beta=\Subst{\Vect{\beta'}}{\Ocons 0{\beta'_l}}l$ we
    have
    \begin{linenomath}
    \begin{align*}
      \Evalstlc(s',\alpha,\Vect\beta,k+3)
      &= \Evalstlc(\State{\Dicelab l{r_l}}{\Stif{M'}{\Loopt\sigma}\Stcons\pi'},
        \alpha,\Vect{\beta},k+2)\\
      &= \Evalstlc(\State{\Num 0}{\Stif{M'}{\Loopt\sigma}\Stcons\pi'},
        \alpha,\Vect{\beta'},k+1)\cdot r_l
        \text{\quad by definition of }\Vect{\beta'}\\
      &= \Evalstlc(\State{M'}{\pi'},\alpha,\Vect{\beta'},k)\cdot r_l\\
      &= \Evalstlc(t',\alpha,\Vect{\beta'},k)\cdot r_l\\
      &\not=\Undef
    \end{align*}
    \end{linenomath}
  \end{itemize}
  This ends the proof of~\Eqref{eq:Lcof_evaldom_struct-2-step-2}.
\end{proof}

To understand the next lemma, remember that
$\Evalstlab(s,\alpha)\in\Mfin L$ so that
$(\Vect r)^{\Evalstlab(s,\alpha)}=\prod_{l\in
  L}r_l^{\Evalstlab(s,\alpha)(l)}$.

\begin{lemma}\label{lemma:Lcof_evaldom_struct-3}
  Let $s\in\Pcfstlab(L)$. Then
  \begin{align*}
    \Evalstlc(\Lcof{\Vect r}(s),\alpha,
    (\Osingle 0^{\Evalstlab(s,\alpha)(l)})_{l\in L})
    =\Proba{\Striplab s}
      (\{\Ocons 0\alpha\})(\Vect r)^{\Evalstlab(s,\alpha)}
  \end{align*}
  for all $\alpha\in\Evdom{\Striplab s}=\Evdomlab{s}$.
\end{lemma}
\begin{proof}
  By definition
  \begin{linenomath}
  \begin{align*}
    \Proba{\Striplab s}(\{\Ocons 0\alpha\})
    =\Evalst(\Striplab s,\alpha)
    \text{\quad and\quad}
    (\Vect r)^{\Evalstlab(s,\alpha)}
    =\prod_{l\in L}r_l^{\Evalstlab(s,\alpha)(l)}
  \end{align*}
  \end{linenomath}
  so we have to prove that
  \begin{linenomath}
  \begin{align*}
    \Evalstlc(\Lcof{\Vect r}(s),\alpha,
    (\Osingle 0^{\Evalstlab(s,\alpha)(l)})_{l\in L})
    =\Evalst(\Striplab s,\alpha)\prod_{l\in L}r_l^{\Evalstlab(s,\alpha)(l)}\,.
  \end{align*}
  \end{linenomath}
  By Lemmas~\ref{lemma:Evaldom-strip}
  and~\ref{lemma:Lcof_evaldom_struct-2} we know that these two
  expressions are defined (that is, all subexpressions are defined) if
  and only if $\alpha\in\Evdomlab s$.
  
  By induction on $n$, we prove that if both expressions
  \begin{linenomath}
  \begin{align*}
    \Evalstlc(\Lcof{\Vect r}(s),\alpha,
    (\Osingle 0^{\Evalstlab(s,\alpha,n)(l)})_{l\in L},n)
    \text{\quad and\quad}
    \Evalst(\Striplab s,\alpha,n)\prod_{l\in L}r_l^{\Evalstlab(s,\alpha,n)(l)}
  \end{align*}
  \end{linenomath}
  are $\not=\Undef$, then they are equal. Assume that the property
  holds for all $p<n$ and let us prove it for $n$, so assume that both
  \begin{align*}
  \Evalstlc(\Lcof{\Vect r}(s),\alpha, (\Osingle
  0^{\Evalstlab(s,\alpha,k)(l)})_{l\in L},k) \text{\quad and\quad}
  \Evalst(\Striplab s,\alpha,k)\prod_{l\in
    L}r_l^{\Evalstlab(s,\alpha,k)(l)}   
  \end{align*}
 are defined, which implies
  $n>0$. We consider the three usual cases as to $s$.
  \begin{itemize}
  \item Assume first that $s=\State{\Let xMN}\pi$ and let
    $t=\State M{\Stlet xN\Stcons\pi}$. We have, using as usual the
    notation $u'=\Lcof{\Vect r}(u)$,
    \begin{linenomath}
    \begin{align*}
      \Evalstlc(s',\alpha, (\Osingle
      0^{\Evalstlab(s,\alpha,n)(l)})_{l\in L},n)
      &= \Evalstlc(t',\alpha, (\Osingle
        0^{\Evalstlab(t,\alpha,n-1)(l)})_{l\in L},n-1)\\
      &= \Evalst(\Striplab t,\alpha,n-1)\prod_{l\in
        L}r_l^{\Evalstlab(t,\alpha,n-1)(l)}\text{\quad by ind.~hyp.}\\
      &= \Evalst(\Striplab s,\alpha,n)\prod_{l\in
        L}r_l^{\Evalstlab(s,\alpha,n)(l)}\,.
    \end{align*}
    \end{linenomath}
  \item Assume now that $s=\State{\Dice r}{\pi}$ so that
    $s'=\State{\Dice r}{\pi'}$. Since
    $\Evalstlab(s,\alpha,n)\not=\Undef$ we must have
    $\alpha=\Ocons i\gamma$ for some $i\in\Eset{0,1}$, and we have
    $\Evalstlab(t,\gamma,n-1)=\Evalstlab(s,\alpha,n)\not=\Undef$ where
    $t=\State{\Num i}{\pi}$. Moreover
    $\Evalstlab(s,\alpha,n)=\Evalstlab(t,\gamma,n-1)$ and hence for
    each $l\in L$ we have
    $\beta_l=\Osingle 0^{\Evalstlab(s,\alpha,n)(l)}=\Osingle
    0^{\Evalstlab(t,\gamma,n-1)(l)}$. We have
    \begin{linenomath}
    \begin{align*}
      &\Evalstlc(s',\alpha, (\Osingle
      0^{\Evalstlab(s,\alpha,n)(l)})_{l\in L},n)\\
      &\hspace{6em}= \Evalstlc(t',\gamma, (\Osingle
        0^{\Evalstlab(t,\alpha,n-1)(l)})_{l\in L},n-1)\cdot\Pneg ir\\
      &\hspace{6em}= \Evalst(\Striplab t,\gamma,n-1)\cdot\Pneg ir
        \prod_{l\in
        L}r_l^{\Evalstlab(t,\gamma,n-1)(l)}\text{\quad by ind.~hyp.}\\
      &\hspace{6em}= \Evalst(\Striplab s,\alpha,n)\prod_{l\in
        L}r_l^{\Evalstlab(s,\alpha,n)(l)}\,.
    \end{align*}
    \end{linenomath}
  \item Assume last that $s=\State{\Mark Ml}\pi$ for some $l\in L$ so
    that $s'=\State{\If{\Dicelab l{r_l}}{M'}{\Loopt\sigma}}{\pi'}$
    where $\sigma$ is the type of $M$.  Let $t=\State M\pi$.  We have
    already noticed that $n>0$; actually, since
    $\Evalstlc(s',\alpha,n)\not=\Undef$, we have $n\geq 3$, see below.
    Moreover $\Evalstlab(t,\alpha,n-1)\not=\Undef$ and
    $\Evalstlab(s,\alpha,n)=\Evalstlab(t,\alpha,n-1)+\Mset l$.
    Let
    $\Vect\beta=(\Osingle 0^{\Evalstlab(s,\alpha,n)(m)})_{m\in L}$, so
    that $\beta_l=\Ocons 0\beta'_l$ and $\beta_m=\beta'_m$ for
    $m\not=l$, where
    $\Vect{\beta'}=(\Osingle 0^{\Evalstlab(t,\alpha,n-1)(m)})_{m\in
      L}$. We have
    \begin{linenomath}
    \begin{align*}
      \Evalstlc(s',\alpha, \Vect\beta,n)
      &=\Evalstlc(\State{\Dicelab l{r_l}}{\Stif{M'}{\Loopt\sigma}\Stcons\pi'},
        \alpha,\Vect\beta,n-1)\\
      &=\Evalstlc(\State{\Num 0}{\Stif{M'}{\Loopt\sigma}\Stcons\pi'},
        \alpha,\Vect{\beta'},n-2)\cdot r_l\\
      &= \Evalstlc(t',\alpha, (\Osingle
        0^{\Evalstlab(t,\alpha,n-1)(m)})_{m\in L},n-3)\cdot r_l\\
      &= \Evalstlc(t',\alpha, (\Osingle
        0^{\Evalstlab(t,\alpha,n-1)(m)})_{m\in L},n-1)\cdot r_l\\
      &\hspace{12em}\text{by monotonicity of step-indexing}\\
      &= \Evalst(\Striplab t,\alpha,n-1)\cdot r_l
        \prod_{m\in
        L}r_m^{\Evalstlab(t,\gamma,n-1)(m)}\text{\quad by ind.~hyp.}\\
      &= \Evalst(\Striplab s,\alpha,n)\prod_{m\in
        L}r_m^{\Evalstlab(s,\alpha,n)(m)}
    \end{align*}
    \end{linenomath}
    by monotonicity of step-indexing, since
    $\Striplab s=\Striplab t$.
  \end{itemize}
  
\end{proof}


 \subsection{The spying translation}\label{sec:sp-trans}
 We consider a last translation, from $\PPCFlab(L)$ to $\PPCF$: let
 $\Vect x$ be an $L$-indexed family of pairwise distinct variables
 (that we identify with the typing context $(x_l:\Tnat)_{l\in L}$). If
 $M\in\PPCFlab(L)$ with $\Tseq{\Gamma}M\sigma$ (assuming that no free
 variable of $M$ occurs in $\Vect x$) we define $\Spy{\Vect x}(M)$
 with $\Tseq{\Gamma,\Vect x}{\Spy{\Vect x}(M)}\sigma$ by induction on
 $M$.  The unique non trivial case is
 $\Spy{\Vect x}(\Mark Ml)=\If{x_l}{\Spy{\Vect x}(M)}{\Loopt\sigma}$
 where $\sigma$ is the type of $M$. As another example, we set
 $\Spy{\Vect x}(\Abst y\tau M)=\Abst y\tau{\Spy{\Vect x}(M)}$ assuming
 of course that $y$ is distinct from all $x_l$'s.

%


 The following lemma is not technically essential, it is simply an
 observation useful to understand better what follows.
\begin{lemma}\label{lemma:spy-sem-zero}
  Let $M\in\PPCFlab(L)$ with $\Tseq{}M\sigma$. If
  $\Vect\rho\in\Mfin\Nat^L=\Mfin{L\times\Nat}$ and
  $a\in\Web{\Tsem\sigma}$ satisfy
  $(\Psem{\Spy{\Vect x}(M)}{\Vect x})_{(\Vect\rho,a)}\not=0$ then
  $\forall l\in L\ \Supp{\rho_l}\subseteq\Eset 0$.
\end{lemma}
Let $f:\Pcoh\Snat^L\to\Pcoh{\Tsem\sigma}$ be the analytic function
induced by $\Psem{\Spy{\Vect x}(M)}{\Vect x}$. This lemma says that
$f(\Vect u)$ (where $\Vect u\in\Pcoh\Snat^L$) depends only on
$(u(l)_0)_{l\in L}\in\Intercc 01^L$, that is, on the $0$-component of
the $u(l)$'s which are probability sub-distributions on $\Nat$.
\begin{proof}
  (Sketch) Simple induction on $M$ considering also open terms: we
  prove that, if $\Tseq{\Gamma}{M}{\sigma}$ with
  $\Gamma=(y_1:\tau_1,\dots,y_k:\tau_k)$, so that
  $\Tseq{\Gamma,\Vect x}{\Spy{\Vect x}(M)}\sigma$, then given
  $\mu_i\in\Mfin{\Web{\Tsem{\tau_i}}}$ for $i=1,\dots,k$,
  $a\in\Web{\Tsem\sigma}$ and $\Vect\rho\in\Mfin\Nat^L$, if
  $(\Psem{\Spy{\Vect x}(M)}{\Gamma,\Vect
    x})_{(\Vect\mu,\Vect\rho,a)}\not=0$ then
  $\forall l\in L\,\forall n\in\Nat\ \rho_l(n)\not=0\Implies
  n=0$. Considering $\Psem{\Spy{\Vect x}(M)}{\Gamma,\Vect x}$ as a
  function
  \begin{linenomath}
  \begin{align*}
    \Psem{\Spy{\Vect x}(M)}{\Gamma,\Vect x}:
    \prod_{i=1}^k\Pcoh{\Tsem{\tau_i}}\times\Pcoh{\Snat}^L\to\Pcoh{\Tsem\sigma}
  \end{align*}
  \end{linenomath}
  (see Section~\ref{sec:pPCF-den-sem-pcoh}) this amounts to showing
  that given $\Vect v\in\prod_{i=1}^k\Pcoh{\Tsem{\tau_i}}$ and
  $\Vect u,\Vect{u'}\in\Pcoh{\Snat}^L$ then
  \begin{linenomath}
  \begin{align*}
    (\forall l\in L
    \ u(l)_0=u'(l)_0)\Implies
    \Psem{\Spy{\Vect x}(M)}{\Gamma,\Vect x}(\Vect v,\Vect u)
    =\Psem{\Spy{\Vect x}(M)}{\Gamma,\Vect x}(\Vect v,\Vect{u'})\,.
  \end{align*}
  \end{linenomath}
  In other words the function
  $\Psem{\Spy{\Vect x}(M)}{\Gamma,\Vect x}(\Vect v,\Vect u)$ of
  $\Vect u$ depends only on the values taken by the $u(l)$'s (the
  components of $\Vect u$) on $0\in\Web{\Snat}=\Nat$. This follows by
  a straightforward induction on $M$, the only ``interesting'' (though
  obvious) case being when $M$ is of shape $\Mark Nl$: in this case we
  have
  \begin{linenomath}
  \begin{align*}
    \Psem{\Spy{\Vect x}(M)}{\Gamma,\Vect x}(\Vect v,\Vect u)
    =u(l)_0\Psem{\Spy{\Vect x}(N)}{\Gamma,\Vect x}(\Vect v,\Vect u)
  \end{align*}
  \end{linenomath}
  since $\Psem{\Loopt\sigma}{}=0$. 
\end{proof}

\begin{lemma}\label{lemma:spy-lc-expression}
  Let $\Vect r\in(\Rational\cap[0,1])^L$ and $M\in\PPCFlab(L)$
  with $\Tseq{}M\tau$. Then
  \begin{align*}
    \Psem{\Spy{\Vect x}(M)}{\Vect x}(\Vect r\,\Base
    0)=\Psem{\Striplc{\Lcof{\Vect r}(M)}}{}\,.
  \end{align*}
\end{lemma}
\begin{proof}(Sketch)
  One proves more generally that given $M$ such that
  $\Tseq\Gamma M\tau$ with
  $\Gamma=(y_1:\sigma_1,\dots,y_k:\sigma_k)$, so that
  $\Tseq{\Gamma,\Vect x}{\Spy{\Vect x}(M)}\tau$, and
  $\Vect v\in\prod_{i=1}^k\Pcoh{\Tsem{\sigma_i}}$, one has
  \begin{linenomath}
  \begin{align*}
    \Psem{\Spy{\Vect x}(M)}{\Gamma,\Vect x}(\Vect u,\Vect r\,\Base
  0)=\Psem{\Striplc{\Lcof{\Vect r}(M)}}{\Gamma}(\Vect u)
  \end{align*}
  \end{linenomath}
  by a simple induction on $M$. The only interesting case is when
  $M=\Mark Nl$:
  \begin{linenomath}
  \begin{align*}
    \Psem{\Spy{\Vect x}(M)}{\Gamma,\Vect x}(\Vect u,\Vect r\,\Base
    0)
    &=(\Vect r\,\Base 0)(l)_0
      \Psem{\Spy{\Vect x}(N)}{\Gamma,\Vect x}(\Vect u,\Vect r\,\Base
    0)\\
    &=r(l)_0
      \Psem{\Spy{\Vect x}(N)}{\Gamma,\Vect x}(\Vect u,\Vect r\,\Base
      0)\\
    &=\Psem{\Striplc{\Lcof{\Vect r}(M)}}{\Gamma}(\Vect u)
  \end{align*}
  \end{linenomath}
  where $\Vect r\,\Base 0=(r_l\Base 0)_{l\in L}\in\Pcoh\Snat^L$, using
  twice the fact that $\Psem{\Loopt\sigma}{}=0$.
\end{proof}

With the constructions and observations accumulated so far, we can
prove the main result of this section. Let $M\in\PPCFlab(L)$ and
$l\in L$. Remember from Section~\ref{sec:ppcf-labels} that
$\Evalstlabfp sl:\Cantorfin\to\Nat$ is the integer r.v.~defined by
$\Evalstlabfp sl(\alpha)=\Evalstlabf s(\alpha)(l)$, the number of
times an $l$-labeled subterm of $M$ has arrived in head position
during the execution of $M$ induced by $\alpha\in\Cantorfin$. So this
r.v.~allows to evaluate the number of execution steps in this
evaluation. For instance if $M$ is obtained by labeling all sub-terms
of a given closed term $N$ of $\PPCF$ of type $\Tnat$ with the same
label $l\in L$, we get an $\Nat$-valued r.v.~which evaluates the
number of execution steps in the evaluation of $N$ that is, the number
of times a subterm of $N$ arrives in head position during this
evaluation (notice that if $N$ is a constant $\Num n$, this number is
$1$).

Given $\mu\in\Mfin L$, we use $\mu\,\Mset 0$ in the proof of the next
result for the element $\rho$ of $\Mfin\Nat^L$ such that
$\rho_l(n)=\mu(l)$ if $n=0$ and $\rho_l(n)=0$ otherwise.

\begin{theorem}\label{th:expect-time-diff}
  Let $M\in\PPCFlab(L)$ with $\Tseq{}M\Tnat$. Then for all $l\in L$
  \begin{linenomath}
    \begin{align*}
    \Expect{\Evalstlabfp{\Inistate M}l\mid\Eventconvz{\Inistate{\Striplab M}}}
    =\frac{\partial\Psem{\Spy{\Vect x}M}{}(\Vect r\Base 0)}{\partial r_l}
    (1,\dots,1)/{\Psem{\Striplab M}{}}_0\in\Realpc\,.
    \end{align*}
  \end{linenomath}
\end{theorem}
\begin{proof}
By Lemma~\ref{lemma:spy-lc-expression},
\begin{align}\label{eq:Psem-Lcof-expr}
  {\Psem{\Striplc{\Lcof{\Vect r}(M)}}{}}_0
  =\sum_{\mu\in\Mfin
  L}(\Psem{\Spy{\Vect x}(M)}{\Vect x})_{(\mu\,\Mset 0,0)}(\Vect
  r)^\mu
\end{align}
By Theorem~\ref{th:pcoh-adequacy}, we have
\begin{linenomath}
  \begin{align*}
  {\Psem{\Striplc{\Lcof{\Vect r}(M)}}{}}_0
  &=\Proba{\Striplc{\Lcof{\Vect r}({\Inistate M})}}
    (\Eventconvz{\Striplc{\Lcof{\Vect r}({\Inistate M})}})\\
  &=\sum_{(\alpha,\Vect\beta)\in\Evdomlc{\Lcof{\Vect r}({\Inistate M})}}
    \Evalstlc(\Lcof{\Vect r}({\Inistate M}),\alpha,\Vect\beta)\quad\text{by Lemma~\ref{lemma:proba-conv-lc}}\\
  &=\sum_{\alpha\in\Evdom{\Striplab {\Inistate M}}}
    \Evalst(\Striplab {\Inistate M},\alpha)\prod_{l\in L}r_l^{\Evalstlab({\Inistate M},\alpha)(l)}
\quad\text{by Lemma~\ref{lemma:Lcof_evaldom_struct-3}}\\
  &=\sum_{\mu\in\Mfin L}\left(\sum_{\Biind{\alpha\in\Ocons 0{\Cantorfin}}{\Evalstlabf{\Inistate M}(\alpha)=\mu}}\Evalstf{\Striplab{\Inistate M}}(\alpha)\right)(\Vect r)^\mu
  \end{align*}
  \end{linenomath}
  and since this holds for all $\Vect r\in(\Rational\cap[0,1])^L$, we
  must have by Equation~\Eqref{eq:Psem-Lcof-expr}, for all
  $\mu\in\Mfin L$,
\begin{linenomath}    
\begin{equation}\label{eq:Evsp-coeff-proba}
  \begin{split}
    (\Psem{\Spy{\Vect x}(M)}{\Vect x})_{(\mu\,\Mset 0,0)}
    &=\sum_{\Biind{\alpha\in\Ocons 0{\Cantorfin}}
      {\Evalstlabf{\Inistate M}(\alpha)=\mu}}
    \Evalstf{\Striplab{\Inistate M}}(\alpha)\\
    &=\Proba{{\Striplab{\Inistate M}}}(\Evalstlabf{\Inistate M}=\mu)
  \end{split}
\end{equation}
\end{linenomath}

Let $l\in L$, we have
\begin{linenomath}
\begin{align*}
  \Expect{\Evalstlabfp{\Inistate M}l}
  &=\sum_{\mu\in\Mfin L}\mu(l)\Proba {\Inistate M}(\Evalstlabf{\Inistate M}=\mu)
    \quad\text{by Equation~\Eqref{eq:esp-evalstlab-multiset}}\\
  &=\sum_{\mu\in\Mfin L}
    \mu(l)(\Psem{\Spy{\Vect x}(M)}{\Vect x})_{(\mu\,\Mset 0,0)}
    \quad\text{by Equation~\Eqref{eq:Evsp-coeff-proba}}\\
  &=\frac{\partial\Psem{\Spy{\Vect x}M}{\Vect x}(\Vect r\Base 0)_0}{\partial r_l}(1,\dots,1)\,.
\end{align*}
\end{linenomath}
Indeed, given $\Vect r\in[0,1]^L$ one has
\begin{linenomath}
\begin{align*}
  \Psem{\Spy{\Vect x}M}{\Vect x}(\Vect r\Base 0)_0=\sum_{\mu\in\Mfin L}
  (\Psem{\Spy{\Vect x}(M)}{\Vect x})_{(\mu\,\Mset 0,0)}\Vect r^\mu 
\end{align*}
\end{linenomath}
and $\frac{\partial \Vect r^\mu}{\partial r_l}(1,...,1)=\mu(l)$,
whence the last equation.
\end{proof}

\begin{example}\label{ex:expect-computation}
  The point of this formula is that we can apply it to algebraic
  expressions of the semantics of the program.  Consider the following
  term $M_q$ (for $q\in\Rational\cap[0,1]$) such that
  $\Tseq{}{M_q}{\Timpl\Tnat\Tnat}$:
  \begin{linenomath}
    \begin{align*}
      M_q=\Fix{\Abst f{\Timpl\Tnat\Tnat}{\Abst x\Tnat{\If{\Dice q}{
      \If{\App fx}{\If{\App fx}{\Num 0}{\Loopt\Tnat}}{\Loopt\Tnat}}
      {\If x{\If x{\Num 0}{\Loopt\Tnat}}{\Loopt\Tnat}}}}}\,,
    \end{align*}
  \end{linenomath}
  we study $\App{M_q}{\Mark{\Num 0}l}$ (for a fixed label
  $l\in\Labels$).  So in this example, ``execution time'' means
  ``number of uses of the parameter $\Num 0$''.  For all
  $v\in\Pcoh\Snat$, we have $\Psem{M_q}{}(v)=\phi_q(v_0)\,\Snum 0$
  where $\phi_q:[0,1]\to[0,1]$ is such that $\phi_q(u)$ is the least
  element of $[0,1]$ which satisfies
  \begin{align*}
    \phi_q(u)=(1-q)\,u^2+q\,\phi_q(u)^2\,.
  \end{align*}
  So
  \begin{linenomath}
    \begin{align*}
      \phi_q(u)=
      \begin{cases}
        \frac{1-\sqrt{1-4q(1-q)u^2}}{2q} & \text{if } q>0\\
        u^2 & \text{if }q=0
      \end{cases}
    \end{align*}
  \end{linenomath}
  the choice between the two solutions of the quadratic equation being
  determined by the fact that the resulting function $\phi_q$ must be
  monotonic in $u$. So by Theorem~\ref{th:pcoh-adequacy} (for
  $q\in(0,1]$)
  \begin{linenomath}
    \begin{align}\label{eq:example-probared}
      \Probared{\App{M_q}{\Num 0}}{\Num 0}=\phi_q(1)=\frac{1-\Absval{2q-1}}{2q}
      =
      \begin{cases}
        1 & \text{if }q\leq 1/2\\
        \frac{1-q}q &\text{if }q>1/2\,.
      \end{cases}
    \end{align}
  \end{linenomath}
  Observe that we have also $\Probared{M_0}{\Num 0}=\phi_0(1)=1$ so
  that Equation~\Eqref{eq:example-probared} holds for all $q\in[0,1]$
  (the corresponding curve is the second one in
  Fig.~\ref{fig:example-expect}).

  Then by Theorem~\ref{th:expect-time-diff} we have
  \begin{linenomath}
    \begin{align*}
      \Expect{\Evalstlabfp{\Inistate{\App{M_q}{\Mark{\Num 0}l}}}l\mid
      \Eventconvz{\Inistate{\App{M_q}{\Num
      0}}}}=\phi'_q(1)/\phi_q(1)\,.
    \end{align*}
  \end{linenomath}
  Since $\phi_q(u)=(1-q)\,u^2+q\,\phi_q(u)^2$ we have
  $\phi'_q(u)=2(1-q)u+2q\phi'_q(u)\phi_q(u)$ and hence
  \begin{align*}
    \phi'_q(1)=2(1-q)/(1-2q\phi_q(1))
  \end{align*}
  so that
  \begin{linenomath}
    \begin{align*}
      \phi'_q(1)/\phi_q(1)=
      \begin{cases}
        2(1-q)/(1-2q) & \text{ if } q<1/2\\
        \infty & \text{ if } q=1/2\\
        2(1-q)/(2q-1) & \text{ if } q>1/2
      \end{cases}
    \end{align*}
  \end{linenomath}
  (using the expression of $\phi_q(1)$ given by
  Equation~\Eqref{eq:example-probared}), see the third curve in
  Fig.~\ref{fig:example-expect}. For $q>1/2$ notice that the
  conditional time expectation \emph{and} the probability of
  convergence decrease when $q$ tends to $1$. When $q$ is very close
  to $1$, $\App{M_q}{\Num 0}$ has a very low probability to terminate,
  but when it does, it uses its argument typically twice. For $q=1/2$
  we have almost sure termination with an infinite expected
  computation time.

  Of course such explicit computations are not always possible. For
  instance, using more occurrences of $\App fx$ we can modify the
  definition of $M_q$ in such a way that computing $\phi_q(u)$ would
  require solving a quintic. Or even we could set
  \begin{linenomath}
    \begin{align*}
      M_q=\Fix{\Abst f{\Timpl\Tnat\Tnat}{\Abst x\Tnat{\If{\Dice q}{
      \If{\App fx}{\If{\App f{\App fx}}{\Num 0}{\Loopt\Tnat}}{\Loopt\Tnat}}
      {\If x{\If x{\Num 0}{\Loopt\Tnat}}{\Loopt\Tnat}}}}}\,,
    \end{align*}
  \end{linenomath}
  and then our solution function $\phi_q$ satisfies
  $\phi_q(u)=(1-q)\,u^2+q\,\phi_q(u)\phi_q(\phi_q(u))$; in such a case
  we cannot expect to have an explicit expression for
  $\phi_q(u)$. Approximating the value of $\phi'_q(1)$ from below is
  possible by performing a finite number of iterations of the fixpoint
  and approximating it from above is a more subtle problem. We could
  also expect to use more efficient approaches typically based on
  Newton's method.
%
\begin{figure}
  \centering
    \begin{tikzpicture}[scale=0.6]
      \begin{axis}
        \addplot[thick, blue,domain=0:1, no markers,samples=500]
        {(1-sqrt(1-x*x)};
      \end{axis}
    \end{tikzpicture}\ 
    \begin{tikzpicture}[scale=0.6]
      \begin{axis}[xtick distance=0.25]
        \addplot[thick, blue,domain=0:0.5, no markers,samples=100]
        {1};
        \addplot[thick, blue,domain=0.5:1, no markers,samples=100]
        {1/x-1};
        \addplot [red, dashed] coordinates {(0.5,0) (0.5,1)};
      \end{axis}
    \end{tikzpicture}\ 
      \begin{tikzpicture}[scale=0.6]
      \begin{axis}[ extra y ticks={2}, xtick distance=0.25]
        \addplot[thick, blue,domain=0:0.45, no markers,samples=100]
        {2*(1-x)/(1-2*x)};
        \addplot[thick, blue,domain=0.55:1, no markers,samples=100]
        {2*x/(2*x-1)};
        \addplot [red, dashed] coordinates {(0.5,0) (0.5,12)};
      \end{axis}
    \end{tikzpicture}
    \caption{Plot of $\phi_{0.5}(u)$ with $u$ on the x-axis (vertical
      slope at $u=1$). Plots of $\phi_q(1)$ and
      $\Expect{\Evalstlabfp{\Inistate{\App{M_q}{\Mark{\Num 0}l}}}l\mid
        \Eventconvz{\Inistate{\App{M_q}{\Num 0}}}}$ with $q$ on the
      x-axis. See Example~\ref{ex:expect-computation}.}
  \label{fig:example-expect}
\end{figure}
\end{example}

\begin{remark} (Connection with relational and coherence semantics.)
  It is possible to interpret terms of $\PPCF$ in $\REL$, the
  relational model of Linear Logic (see for
  instance~\cite{BucciarelliEhrhard99}). In this model each type
  $\sigma$ is interpreted as a set $\Tsemrel\sigma$:
  \begin{linenomath}
  \begin{align*}
    \Tsemrel\Snat&=\Nat\\
    \Tsemrel{\Timpl\sigma\tau}&=\Mfin{\Tsemrel\sigma}\times\Tsemrel\tau
  \end{align*}
  \end{linenomath}
  so that for each type $\sigma$ we have
  \begin{linenomath}
  \( 
    \Tsemrel\sigma=\Web{\Tsem{\sigma}}\,.
    \) 
    \end{linenomath}
  If $\Tseq\Gamma M\tau$ with
  $\Gamma=(x_1:\sigma_1,\dots,x_k:\sigma_k)$ then
  \begin{linenomath}
  \begin{align*}
    \Psem M\Gamma\in\REL\left(\prod_{i=1}^k\Mfin{\Tsemrel{\sigma_i}},
    \Tsemrel\tau\right)
    =\Part{\prod_{i=1}^k\Mfin{\Tsemrel{\sigma_i}}\times
    \Tsemrel\tau}\,.
  \end{align*}
\end{linenomath}
This semantics is ``qualitative'' in the sense that a point can
only belong or not belong to the interpretation of a term whereas the
$\PCOH$ semantics is quantitative in the sense that the interpretation
of the same term also provides a coefficient $\in\Realp$ for this
point. We explain shortly the connection between the two models.
To this end we describe first the relational model. One of the
shortest ways to do so is by means of the ``intersection typing
system'' given in Fig.~\ref{fig:PPCF-rel-intertypes} where we use the
following conventions:
  \begin{itemize}
  \item $\Phi,\Phi_0\cdots$ are \emph{semantic contexts} of shape
    $\Phi=(x_1:\mu_1:\sigma_1,\dots,x_k:\mu_k:\sigma_k)$ where the
    $x_i$'s are pairwise distinct variables and
    $\mu_i\in\Mfin{\Tsemrel{\sigma_i}}$;
  \item if $\Phi=(x_1:\mu_1:\sigma_1,\dots,x_k:\mu_k:\sigma_k)$ is
    such a semantic context then
    $\Underc\Phi=(x_1:\sigma_1,\dots,x_k:\sigma_k)$ is the
    underlying typing context;
  \item given a typing context
    $\Gamma=(x_1:\sigma_1,\dots,x_k:\sigma_k)$, $\Zeroc\Gamma$
    stands for the semantic context
    $\Zeroc\Gamma=(x_1:\Mset{}:\sigma_1,\dots,x_k:\Mset{}:\sigma_k)$;
  \item given semantic contexts
    $\Phi_j=(x_1:\mu^j_1:\sigma_1,\dots,x_k:\mu^j_k:\sigma_k)$ for
    $j=1,\dots,n$ which have all the same underlying typing context
    $\Gamma=(x_1:\sigma_1,\dots,x_k:\sigma_k)$, $\sum_{j=1}^n\Phi_j$
    stands for the semantic context
    $(x_1:\sum_{j=1}^n\mu^j_1:\sigma_1,\dots,x_k:\sum_{j=1}^n\mu^j_k:\sigma_k)$
    whose underlying typing context is $\Gamma$ ($\Zeroc\Gamma$ can
    be seen as the case $n=0$ of this construct with the slight
    problem that $\Gamma$ cannot be derived from the $\Phi_j$'s in
    that case since there are none, whence the special construct
    $\Zeroc\Gamma$).
  \end{itemize}
  Then, assuming that $\Tseq{(x_i:\sigma_i)_{i=1}^k} M\tau$, this
  typing system is such that, given
  $\Vect\mu\in\prod_{i=1}^k\Mfin{\Tsemrel{\sigma_i}}$ and
  $a\in\Tsemrel\tau$, one has $(\Vect\mu,a)\in\Psemrel M\Gamma$ iff
  the judgment
  \begin{linenomath}
  \begin{align*}
    \Sseq{(x_i:\mu_i:\sigma_i)_{i=1}^k}{M}{a}{\tau}
  \end{align*}
  \end{linenomath}
  is derivable.  The interpretation of a term in $\REL$ is simply the
  support of its interpretation in $\PCOH$:
  \begin{linenomath}
  \begin{align*}
    (\Vect\mu,a)\in\Psemrel M\Gamma\Equiv(\Psem M\Gamma)_{\Vect\mu,a}\not=0
  \end{align*}
  \end{linenomath}
  as soon as all occurrences of $\Dice r$ in $M$ are such that
  $r\notin\Eset{0,1}$ (occurrences of $\Dice 0$ and $\Dice 1$ can be
  replaced by $\Num 1$ and $\Num 0$ respectively without changing the
  semantics of $M$). This is easy to prove by a simple induction on
  $M$.

  Since~\cite{BucciarelliEhrhard99,Boudes11} we know that Girard's
  coherence space semantics can be modified as follows: a
  \emph{non-uniform coherence space} is a triple
  $X=(\Web X,\scoh_X,\sincoh_X)$ where $\Web X$ is an at most
  countable set (the web of $X$) and $\scoh_X$, $\sincoh_X$ are
  disjoint binary symmetric relations on $\Web X$ called \emph{strict
    coherence} and \emph{strict incoherence} but contrarily to
  ordinary coherence spaces we can have $a\scoh_X a$ or $a\sincoh_X a$
  for some $a\in\Web X$. These objects can be organized into a
  categorical model of classical linear logic $\NUCS$ whose associated
  Kleisli cartesian closed category is a model of $\PCF$ (that is,
  $\PPCF$ without the $\Dice r$ construct).  Contrarily to what
  happens with Girard's coherence spaces\footnote{Indeed in Girard's
    coherence spaces, the web of $\Excl X$ is the set of all finite
    cliques, or all finite multicliques of elements of $\Web X$ (there
    are two versions of this exponential), hence this web \emph{depends on the
    coherence relation $\scoh_X$}. This is no more the case with
    non-uniform coherence spaces and this is the most important
    difference between the two models.}, we have
  \begin{linenomath}
  \begin{align*}
    \Web{\Tsemcoh\sigma}=\Tsemrel\sigma=\Web{\Tsem\sigma}\,.
  \end{align*}
  \end{linenomath}

  Moreover given a $\PCF$ term $M$ such that $\Tseq{}M\tau$, the set
  $\Tsemcoh M$, which is a clique of the non-uniform coherence space
  $\Tsemcoh\tau$ ---~meaning that
  $\forall a,a'\in\Psem M{}\ \neg (a\sincoh_{\Tsemcoh\tau}a')$~---,
  satisfies
  \begin{linenomath}
    \begin{align*}
      \Psemcoh M{}
      =\Psemrel M{}=\Eset{a\in\Web{\Tsem\tau}\St{\Psem M{}}_a\not=0}\,.
    \end{align*}
  \end{linenomath}
  In other words, the interpretation of a $\PCF$ term in $\NUCS$ is
  \emph{exactly the same} as its interpretation in the basic model
  $\REL$. So what is the point of the model $\NUCS$?  It teaches us
  something we couldn't see in $\REL$: $\Psemrel M{}$ is a clique in
  the non-uniform coherence space associated with its type in $\NUCS$.
  
  In the model $\NUCS$, the interpretation of the object of integers
  $\Snat^\NUCS=\Tsemcoh\Tnat$ satisfies $\Web{\Snat^\NUCS}=\Nat$,
  $n\sincoh n'$ if $n\not=n'$ and
  $\neg(n\scoh n)\wedge\neg(n\sincoh n)$ for all $n\in\Nat$. In the
  model $\NUCS$ at least two exponentials are available; the free one
  is characterized in~\cite{Boudes11}. With this exponential,
  $\Excl{\Snat^\NUCS}$ has $\Mfin\Nat$ as web and
  \begin{itemize}
  \item $\mu\sincoh\mu'$ if
    $\exists n\in\Supp\mu,\,n'\in\Supp{\mu'}\ n\not=n'$
  \item $\mu\scoh\mu'$ if $\Supp\mu\cup\Supp{\mu'}$ has at most one
    element and $\mu\not=\mu'$.
  \end{itemize}
  Notice in particular that $\Mset{0,1}\sincoh\Mset{0,1}$ in
  $\Excl{\Snat^\NUCS}$.

  Let $M\in\PCFlab(L)$ (that is $M\in\PPCFlab(L)$ and $M$ contains no
  instances of $\Dice r$) such that $\Tseq{}M\Tnat$ so that
  $\Spy{\Vect x}(M)\in\PCF$ satisfies
  $\Tseq{(x_l:\Tnat)_{l\in L}}{\Spy{\Vect x}(M)}\Tnat$ where
  $\Vect x=(x_l:\Tnat)_{l\in L}$ is a list of pairwise distinct
  variables. Then $\Psemcoh{\Spy{\Vect x}(M)}{\Vect x}$ is a clique of
  the non-uniform coherence space
  $X=\Limpl{\Excl{\Snatcoh}\ITens\cdots\ITens\Excl{\Snatcoh}}\Snatcoh$
  (one occurrence of $\Excl{\Snatcoh}$ for each element of $L$). If
  $(\Vect\mu,n)\in\Psemcoh{\Spy{\Vect x}(M)}{\Vect x}$ then we know
  that each $\mu(l)\in\Mfin\Nat$ satisfies
  $\Supp{\mu(l)}\subseteq\Eset 0$ (see
  Lemma~\ref{lemma:spy-sem-zero}). And since the set
  $\Psemcoh{\Spy{\Vect x}(M)}{\Vect x}$ is a clique in $X$ and in view
  of the characterization above of the coherence relation of
  $\Excl\Snatcoh$, this set contains at most one element. When it is
  empty, this means that the execution of $M$ does not terminate. When
  it is a singleton $\Eset{(\Vect\mu,n)}$, the execution of $M$
  terminates with value $\Num n$, using $\mu(l)(0)$ times the
  $l$-labeled subterms of $M$. This can be understood as a version of
  the denotational characterization of execution time developed
  in~\cite{DeCarvalho09,DeCarvalho18}, which is based on the model
  $\REL$.

  From the viewpoint of the denotational interpretation in $\PCOH$,
  non-uniform coherence spaces tell us that, for non-probabilistic
  labeled closed terms $M\in\PCFlab(L)$ of type $\Tnat$, the power
  series $\Psem{\Spy{\Vect x}(M)}{}$ has at most one monomial whose
  degree reflects the number of times its labeled subterms are used
  during its (deterministic) execution. Remember indeed that
  \begin{align*}
    \Psemrel{\Spy{\Vect x}(M)}{\Vect x}
    =\Eset{(\Vect\mu,n)
    \St(\Psem{\Spy{\Vect x}(M)}{\Vect x})_{\Vect\mu,n}\not=0}\,.
  \end{align*}
  When $M\in\PPCFlab(L)$, our Theorem~\ref{th:expect-time-diff} is
  a ``smooth'' version of this property.
  \begin{figure}
    \centering
    \footnotesize{
  \begin{center}    
    \AxiomC{$\mu_j=
      \begin{cases}
        \Mset{a} & \text{if }j=i\\
        \Mset{} & \text{otherwise}
      \end{cases}
    $}
    \UnaryInfC{$\Sseq{(x_j:\mu_j:\sigma_j)_{j=1}^k}{x_i}{a}{A}$}
    \DisplayProof
    \quad
    \AxiomC{}
    \UnaryInfC{$\Sseq{\Zeroc\Gamma}{\Num n}{n}{\Tnat}$}
    \DisplayProof
  \end{center}
  \begin{center}
    \AxiomC{$\Sseq\Phi Mn\Tnat$}
    \UnaryInfC{$\Sseq\Phi {\Succ M}{n+1}\Tnat$}
    \DisplayProof
    \quad
    \AxiomC{$\Sseq\Phi M0\Tnat$}
    \UnaryInfC{$\Sseq\Phi {\Pred M}{0}\Tnat$}
    \DisplayProof
    \quad
    \AxiomC{$\Sseq\Phi M{n+1}\Tnat$}
    \UnaryInfC{$\Sseq\Phi {\Pred M}{n}\Tnat$}
    \DisplayProof
  \end{center}
  \begin{center}
    \AxiomC{$r\in\Interoc01\cap\Rational$}
    \UnaryInfC{$\Sseq{\Zeroc\Gamma}{\Dice r}{0}{\Tnat}$}
    \DisplayProof
    \quad
    \AxiomC{$r\in\Interco 01\cap\Rational$}
    \UnaryInfC{$\Sseq{\Zeroc\Gamma}{\Dice r}{1}{\Tnat}$}
    \DisplayProof
  \end{center}
  \begin{center}
    \AxiomC{$\Sseq\Phi Mn\Tnat$}
    \UnaryInfC{$\Sseq{\Phi}{\Succ M}{n+1}{\Tnat}$}
    \DisplayProof
    \quad
    \AxiomC{$\Sseq{\Phi_0}M0\Tnat$}
    \AxiomC{$\Sseq{\Phi_1}PaA$}
    \AxiomC{$\Tseq{\Gamma}QA$
      where $\Gamma=\Underc{\Phi_0}=\Underc{\Phi_1}$}
    \TrinaryInfC{$\Sseq{\Phi_0+\Phi_1}{\If MPQ}{a}{A}$}
    \DisplayProof
  \end{center}
  \begin{center}
    \AxiomC{$\Sseq{\Phi_0}{M}{n}{\Tnat}$}
    \AxiomC{$\Sseq{\Phi_1,x:\Mset{n,\dots,n}:\Tnat}Na\sigma$}
    \AxiomC{$\Underc{\Phi_0}=\Underc{\Phi_1}$}
    \TrinaryInfC{$\Sseq{\Phi_0+\Phi_1}{\Let xMN}{a}{\sigma}$}
    \DisplayProof
  \end{center}
  \begin{center}
    \AxiomC{$\Sseq{\Phi_0}M{n+1}\Tnat$}
    \AxiomC{$\Sseq{\Phi_2}QaA$}
    \AxiomC{$\Tseq{\Gamma}PA$
      where $\Gamma=\Underc{\Phi_0}=\Underc{\Phi_2}$}
    \TrinaryInfC{$\Sseq{\Phi_0+\Phi_2}{\If MPQ}{a}{A}$}
    \DisplayProof
  \end{center}
  \begin{center}
    \AxiomC{$\Sseq{\Phi,x:\mu:A}{M}{b}{B}$}
    \UnaryInfC{$\Sseq{\Phi}{\Abst xAM}{(\mu,b)}{\Timpl AB}$}
    \DisplayProof
  \end{center}
  \begin{center}
    \AxiomC{$\Sseq{\Phi_0}{M}{(\Mset{\List a1n},b)}{\Timpl AB}$}
    \AxiomC{$\Sseq{\Phi_i}{P}{a_i}{A}$ and
      $\Underc{\Phi_i}=\Underc{\Phi_0}$ for $i=1,\dots,n$}
    \BinaryInfC{$\Sseq{\sum_{i=0}^n\Phi_i}{\App MP}{b}{B}$}
    \DisplayProof
  \end{center}
  \begin{center}
    \AxiomC{$\Sseq{\Phi_0}{M}{(\Mset{\List a1n},a)}{\Timpl AA}$}
    \AxiomC{$\Sseq{\Phi_i}{\Fix M}{a_i}{A}$ and
      $\Underc{\Phi_i}=\Underc{\Phi_0}$ for $i=1,\dots,n$}
    \BinaryInfC{$\Sseq{\sum_{i=0}^n\Phi_i}{\Fix M}{a}{A}$}
    \DisplayProof
  \end{center}
      }
      \caption{Relational interpretation of $\PPCF$ as an
        intersection typing system}
    \label{fig:PPCF-rel-intertypes}
  \end{figure}
\end{remark}


\section{Differentials and distances}

In this second part of the paper, we also use differentiation in the
category $\PCOH$, but contrarily to what we did when relating derivatives with
execution time ---~we used only derivatives of first order
functions~--- we will now consider also the ``derivatives'' (in that case
one rather uses the word ``differentials'') of higher order
functions. This is possible thanks to the fact that, even at higher
order, our functions are analytic in some sense. 

\subsection{Order theoretic characterization of PCSs}

The following simple lemma will be useful in the sequel. It
is proven in~\cite{Girard04a} in a rather sketchy way, we provide here a
detailed proof for further reference.  We say that a partially
ordered set $S$ is $\omega$-complete if any increasing sequence of
elements of $S$ has a least upper bound in $S$.

\begin{lemma}\label{lemma:PCS-as-closed-set}
  Let $I$ be a countable set and let $P\subseteq\Realpto I$.  Then
  $(I,P)$ is a probabilistic coherence space iff the following
  properties hold (equipping $P$ with the product order).
  \begin{enumerate}
  \item $P$ is downwards closed and closed under barycentric combinations
  \item $P$ is $\omega$-complete
  \item \label{item:pcs-bdd-hyp}and for all $a\in I$ there is
    $\epsilon>0$ such that $\epsilon\Bcanon a\in P$ and
    $P_a\subseteq[0,1/\epsilon]$.
  \end{enumerate}
\end{lemma}
\begin{proof}
  The $\Implies$ implication is easy (see~\cite{DanosEhrhard08}), we
  prove the converse, which uses the Hahn-Banach theorem in finite
  dimension. Notice first that by condition~(\ref{item:pcs-bdd-hyp})
  we have $\Orth P,\Biorth P\subseteq\Realpto I$.

  Let $P\subseteq\Realpto I$ satisfying the three conditions above and
  let us prove that $\Biorth P\subseteq P$, that is, given
  $y\in\Realpto I\setminus P$, we must prove that $y\notin\Biorth P$,
  that is, we must exhibit a $x'\in\Orth P$ such that $\Eval
  y{x'}>1$. We first show that we can assume that $I$ is finite.

  Given $J\subseteq I$ and $z\in\Realpto I$, let $\FamRestr zJ$ be the
  element of $\Realpto I$ which takes value $z_j$ for $j\in J$ and $0$
  for $j\notin J$. Then $y$ is the lub of the increasing sequence
  $(\FamRestr y{\{\List i1n\}})_{n\in\Nat}$ (where $i_1,i_2,\dots$ is
  any enumeration of $I$) and hence there must be some $n\in\Nat$ such
  that $\FamRestr y{\{\List i1n\}}\notin P$ by the assumption that $P$
  is $\omega$-complete.  Therefore it suffices to prove the result for
  $I$ finite, which we assume.  Let
  $Q=\{x\in\Realto I\St (\Absval{x_i})_{i\in I}\in P\}$ which is a
  convex subset of $\Realto I$ by the assumption that $P$ is convex
  and downwards-closed.

  Let $t_0=\sup\{t\in\Realp\St ty\in P\}$. By our assumption that $P$
  is $\omega$-complete, we have $t_0y\in P$ and hence $t_0<1$ since
  $y\notin P$. Let $h:\Real y=\{ty\St t\in\Real\}\to\Real$ be defined
  by $h(ty)=t/t_0$ ($t_0\not=0$ by our
  assumption~(\ref{item:pcs-bdd-hyp}) about $P$ and because $I$ is
  finite). Let $q:\Realto I\to\Realp$ be the gauge of $Q$, which is
  the semi-norm given by
\(
q(z)=\inf\{\epsilon >0\St z\in\epsilon Q\}
\).
It is actually a norm by our assumptions on $P$. Observe that
$h(z)\leq q(z)$ for all $z\in\Real y$: this boils down to showing that
$t\leq t_0 q(ty)=\Absval tt_0q(y)$ for all $t\in\Real$ which is clear
since $t_0q(y)=1$ by definition of these numbers.  Hence, by the
Hahn-Banach Theorem\footnote{Here is one of the many versions of the
  Hahn-Banach Theorem: let $E$ be a $\Real$-vector space, $F$ a
  subspace of $E$, $f:F\to\Real$ a linear map, $p:E\to\Realp$ a
  seminorm such that $\Absval f\leq p$ on $F$. Then there is a
  $g:E\to\Real$, linear and extending $f$ and such that
  $\Absval g\leq p$ on $E$.}, there exists a linear
$l:\Realto I\to\Real$ which is upper-bounded by $q$ and coincides with
$h$ on $\Real y$. Let $y'\in\Realto I$ be such that $\Eval z{y'}=l(z)$
for all $z\in\Realto I$ (using again the finiteness of $I$). Let
$x'\in\Realpto I$ be defined by $x'_i=\Absval{y'_i}$. It is clear that
$\Eval y{x'}>1$: since $y\in\Realpto I$ we have
$\Eval y{x'}\geq\Eval y{y'}=l(y)=h(y)=1/t_0>1$. Let
$N=\{i\in I\St y'_i<0\}$.  Given $z\in P$, let $\bar z\in\Realto I$ be
given by $\bar z_i=-z_i$ if $i\in N$ and $\bar z_i=z_i$
otherwise. Then $\Eval z{x'}=\Eval{\bar z}{y'}=l(\bar z)\leq 1$ since
$\bar z\in Q$ (by definition of $Q$ and because $z\in P$). It follows
that $x'\in\Orth P$.
\end{proof}

\subsection{Local PCS and derivatives}\label{sec:loca-PCS-diff}

Given a cone $P$ (see Section~\ref{sec:basics-PCSs} for the
definition) and $x\in\Cuball P$, we define the \emph{local cone at
  $x$} as the set
\(
\Cloc Px=\{u\in P\St\exists\epsilon>0\ x+\epsilon u\in\Cuball P\}
\).
Equipped with the algebraic operations inherited from $P$, this set is
clearly a $\Realp$-semi-module. We equip it with the following norm:
\(
  \Normsp u{\Cloc Px}=\inf\{\epsilon^{-1}\St\epsilon>0\text{ and }
  x+\epsilon u\in\Cuball P\}
\)
and then it is easy to check that $\Cloc Px$ is indeed a cone. It is
reduced to $0$ exactly when $x$ is maximal in $\Cuball P$. In that
case one has $\Normsp xP=1$ but notice that the converse is not true
in general.

We specialize this construction to PCSs.  Let $X$ be a PCS and let
$x\in\Pcoh X$. We define a new PCS $\Locpcs Xx$ as follows. First we
set
\(
  \Web{\Locpcs Xx}=\{a\in\Web X\St\exists\epsilon>0\
  x+\epsilon\Bcanon a\in\Pcoh X\}
\)
and then
\(
  \Pcohp{\Locpcs Xx}=\{u\in\Realpto{\Web{\Locpcs Xx}}\St x+u\in\Pcoh X\}
\).
There is a slight abuse of notation here: $u$ is not an element of
$\Realpto{\Web X}$, but we consider it as such by simply extending it
with $0$ values to the elements of $\Web X\setminus\Web{\Locpcs Xx}$.
Observe also that, given $u\in\Pcoh X$, if $x+u\in\Pcoh X$, then we
\emph{must have} $u\in\Pcohp{\Locpcs Xx}$, in the sense that $u$
necessarily vanishes outside $\Web{\Locpcs Xx}$.
It is clear that $(\Web{\Locpcs Xx},\Pcohp{\Locpcs Xx})$ satisfies the
conditions of Lemma~\ref{lemma:PCS-as-closed-set} and therefore
$\Locpcs Xx$ is actually a PCS, called the \emph{local PCS of $X$ at
  $x$}.

Let $t\in\Kl\PCOH(X,Y)$ and let $x\in\Pcoh X$. Given
$u\in\Pcohp{\Locpcs Xx}$, we know that $x+u\in\Pcoh X$ and hence we can
compute $\Klfun t(x+u)\in\Pcoh Y$:
\begin{linenomath}
  \begin{align*}
    \Klfun t(x+u)_b=\sum_{\mu\in\Web{\Excl X}}t_{\mu,b}(x+u)^\mu
    =\sum_{\mu\in\Web{\Excl X}}t_{\mu,b}\sum_{\nu\leq\mu}
    \Binom \mu\nu x^{\mu-\nu}u^\nu\,.
  \end{align*}
\end{linenomath}
Upon considering only the $u$-constant and the $u$-linear parts of this
summation (and remembering that actually $u\in\Pcohp{\Locpcs Xx}$), we
get
\begin{linenomath}
\begin{align*}
  \Klfun t(x)+
  \sum_{a\in\Web X}u_a\sum_{\mu\in\Web{\Excl X}}(\mu(a)+1)
  t_{\mu+\Mset a,b}x^\mu\Base b
  \leq\Klfun t(x+u)\in\Pcoh Y\,.
\end{align*}
\end{linenomath}
Given $a\in\Web{\Locpcs Xx}$ and $b\in\Web{\Locpcs Y{\Klfun t(x)}}$, we set
\begin{linenomath}
  \begin{align*}
  \Deriv t(x)_{a,b}=\sum_{\mu\in\Web{\Excl X}}(\mu(a)+1)t_{\mu+\Mset a,b}x^\mu
  \end{align*}
\end{linenomath}
and we have proven that actually
\begin{linenomath}
  \begin{align*}
    \Deriv t(x)\in\PCOH(\Locpcs Xx,\Locpcs Y{\Klfun t(x)})\,.
  \end{align*}
\end{linenomath}
By definition, this linear morphism $\Deriv t(x)$ is the
\emph{derivative (or differential, or Jacobian) of $t$ at
  $x$}\footnote{But unlike our models of Differential LL, this
  derivative is only defined locally; this is slightly reminiscent of
  what happens in differential geometry.}.
It is uniquely characterized by the fact that, for all $x\in\Pcoh X$
and $u\in\Pcoh{\Locpcs Xx}$, we have
\begin{align}
  \Fun t(x+u)=\Fun t(x)+\Matappa{\Deriv t(x)}u+\Rem t(x,u)
\end{align}
where $\Rem t$ is a power series in $x$ and $u$ whose all terms have
global degree $\geq 2$ in $u$.

\begin{example}
  Consider the case where $Y=\Excl X$ and
  $t=\Diracm=\Id_{\Excl X}\in\Kl\PCOH(X,\Excl X)$, so that
  $\Klfun\Diracm(x)=\Prom x$. Given $a\in\Web{\Locpcs Xx}$ and
  $\nu\in\Web{\Locpcs{\Excl X}{\Prom x}}$, we have
  \begin{linenomath}
    \begin{align*}
      \Deriv\Diracm(x)_{a,\nu}
      =\sum_{\mu\in\Web{\Excl X}}(\mu(a)+1)\Diracm_{\mu+\Mset a,\nu}x^\mu
      =
      \begin{cases}
        0 & \text{ if }\nu(a)=0\\
        \nu(a)x^{\nu-\Mset a} & \text{ if }\nu(a)>0\,.
      \end{cases}
    \end{align*}
  \end{linenomath}
\end{example}
%

We know that
$\Deriv\Diracm(x)\in\Pcohp{\Limpl{\Locpcs Xx}{\Locpcs{\Excl X}{\Prom
      x}}}$
so that $\Deriv\Diracm(x)$ is a ``local version'' of DiLL's
codereliction~\cite{Ehrhard18}. Observe for instance that
$\Deriv\Diracm(0)$ satisfies
$\Deriv\Diracm(0)_{a,\nu}=\Kronecker{\nu}{\Mset a}$ and therefore
coincides with the ordinary definition of codereliction.

\begin{proposition}[Chain Rule]\label{prop:chain-rule}
  Let $s\in\Kl\PCOH(X,Y)$ and $t\in\Kl\PCOH(Y,Z)$. Let $x\in\Pcoh X$
  and $u\in\Pcoh{\Locpcs Xx}$. Then we have
  $\Matappa{\Deriv{(t\Comp s)}(x)}u=\Matappa{\Deriv t(\Fun
    s(x))}{\Matappa{\Deriv s(x)}u}$.
\end{proposition}
\begin{proof}
It suffices to write
\begin{linenomath}
  \begin{align*}
  \Fun{(t\Comp s)}(x+u)
  &= \Fun t(\Fun s(x+u))
  = \Fun t(\Fun s(x)+\Matappa{\Deriv s(x)}u+\Rem s(x,u))\\
  &= \Fun t(\Fun s(x))+\Matappa{\Deriv t(\Fun s(x))}{(\Matappa{\Deriv s(x)}u
    +\Rem s(x,u)))}+\Rem t(\Fun s(x),\Matappa{\Deriv s(x)}u
    +\Rem s(x,u))\\
  &= \Fun t(\Fun s(x))+\Matappa{\Deriv t(\Fun s(x))}{(\Matappa{\Deriv s(x)}u)}
  +\Matappa{\Deriv t(\Fun s(x))}{(\Rem s(x,u))}
    +\Rem t(\Fun s(x),\Matappa{\Deriv s(x)}u
    +\Rem s(x,u))
  \end{align*}
\end{linenomath}
by linearity of $\Deriv t(\Fun s(x))$ which proves our contention by
the observation that, in the power series
$\Matappa{\Deriv t(\Fun s(x))}{(\Rem s(x,u))} +\Rem t(\Fun
s(x),\Matappa{\Deriv s(x)}u +\Rem s(x,u))$, $u$ appears with global
degree $\geq 2$ by what we know on $\Rem s$ and $\Rem t$.
\end{proof}



\subsection{Glb's, lub's and distance}
Since we are working with probabilistic coherence spaces, we could
deal directly with families of real numbers and define these
operations more concretely. We prefer not to do so to have a more
canonical presentation which can be generalized to cones such as those
considered in~\cite{EhrhardPaganiTasson18,Ehrhard20}. Given a PCS $X$,
remember that $\Normsp\_X$ denotes the norm $\Normsp\_{\Pcohc X}$ of
the associated cone, see Section~\ref{sec:basics-PCSs}.

Given $x,y\in\Pcoh X$, observe that $x\Infi y\in\Pcoh X$, where
$(x\Infi y)_a=\min(x_a,y_a)$, and that $x\Infi y$ is the glb of $x$
and $y$ in $\Pcoh X$ (with its standard ordering). It follows that $x$
and $y$ have also a lub $x\Supi y\in\Pcohc X$ which is given by
$x\Supi y=x+y-(x\Infi y)$ (and of course
$(x\Supi y)_a=\max(x_a,y_a)$).

Let us prove that $x+y-(x\Infi y)$ is actually the lub of $x$ and
$y$. First, $x\leq x+y-(x\Infi y)$ simply because $x\Infi y\leq y$.
Next, let $z\in\Pcohc X$ be such that $x\leq z$ and $y\leq z$. We must
prove that $x+y-(x\Infi y)\leq z$, that is
$x+y\leq z+(x\Infi y)=(z+x)\Infi(z+y)$, which is clear since
$x+y\leq z+x,z+y$. We have used the fact that $+$ distributes over
$\Infi$ so let us prove this last fairly standard property:
$z+(x\Infi y)=(z+x)\Infi(z+y)$. The ``$\leq$'' inequation is obvious
(monotonicity of $+$) so let us prove the converse, which amounts to
$x\Infi y\geq(z+x)\Infi(z+y)-z$ (observe that indeed that
$z\leq (z+x)\Infi(z+y)$). This in turn boils down to proving that
$x\geq(z+x)\Infi(z+y)-z$ (and similarly for $y$) which results
from $x+z\geq(z+x)\Infi(z+y)$ and we are done.

We define the distance between $x$ and $y$ by
\begin{linenomath}
\begin{align*}
  \Distsp Xxy=\Normsp{x-(x\Infi y)}X+\Normsp{y-(x\Infi y)}X\,.
\end{align*}
\end{linenomath}
The only non obvious fact to check to prove that this is actually a
distance is the triangular inequality, so let $x,y,z\in\Pcoh X$. We
have
$x-(x\Infi z)\leq x-(x\Infi y\Infi z)=x-(x\Infi y)+(x\Infi y)-(x\Infi
y\Infi z)$ and hence
\begin{linenomath}
\begin{align*}
  \Normsp{x-(x\Infi z)}X\leq\Normsp{x-(x\Infi y)}X+\Normsp{(x\Infi y)-(x\Infi
    y\Infi z)}X\,.
\end{align*}
\end{linenomath}
Now we have $(x\Infi y)\Supi(y\Infi z)\leq y$, that is
$(x\Infi y)+(y\Infi z)-(x\Infi y\Infi z)\leq y$, that is
$(x\Infi y)-(x\Infi y\Infi z)\leq y-(y\Infi z)$. It follows that
\begin{linenomath}
\begin{align*}
  \Normsp{x-(x\Infi z)}X\leq\Normsp{x-(x\Infi y)}X+\Normsp{y-(y\Infi z)}X
\end{align*}
\end{linenomath}
and symmetrically
\begin{linenomath}
\begin{align*}
\Normsp{z-(x\Infi z)}X\leq\Normsp{z-(z\Infi y)}X+\Normsp{y-(y\Infi x)}X
\end{align*}
\end{linenomath}
and summing up we get, as expected
\(
\Distsp Xxz\leq\Distsp Xxy+\Distsp Xyz
\).

\begin{remark}\label{rk:cone-distance}
  In a cone $P$, glb's do not necessarily exist; we can nevertheless
  define a distance as follows:
  \begin{align*}
    \Distsp Pxy=\inf\{\Normsp{x-z}P+\Normsp{y-z}P
    \St z\in P\text{ and }z\leq x,y\}
  \end{align*}
  and it is possible to prove that, equipped with this distance, $P$
  is always Cauchy-complete. Of course if $P=\Pcohc X$ this distance
  coincides with the distance defined above.
\end{remark}

\subsection{A Lipschitz property}
Using the differential of Section~\ref{sec:loca-PCS-diff}, we prove
that all morphisms of $\Kl\PCOH$ satisfy a Lipschitz property, with a
coefficient which cannot be upper bounded on the whole domain.

First of all, observe that, if $w\in\Pcohcp{\Limpl XY}$ and
$x\in\Pcohc X$, we have
\begin{align*}
  \Normsp{\Matappa wx}Y\leq\Normsp w{\Limpl XY}\,\Normsp xX\,.
\end{align*}
Indeed if $\Normsp w{\Limpl XY}\not=0$ and $\Normsp xX\not=0$ we have
$\frac{w}{\Normsp w{\Limpl XY}}\in\Pcohp{\Limpl XY}$ and
$\frac{x}{\Normsp xX}\in\Pcoh X$, therefore
$\Matappa{\frac{w}{\Normsp w{\Limpl XY}}}{\frac{x}{\Normsp xX}}\in\Pcoh
Y$ and our contention follows. And if $\Normsp w{\Limpl XY}=0$ or
$\Normsp xX=0$ the inequation is obvious since then $\Matappa wx=0$.

Let $p\in\Interco01$.
If $x\in\Pcoh X$ and $\Normsp xX\leq p$, observe that, for any
$u\in\Pcoh X$, one has
\begin{align*}
  \Normsp{x+(1-p)u}X\leq\Normsp xX+(1-p)\Normsp uX\leq 1  
\end{align*}
and hence
$(1-p)u\in\Pcohp{\Locpcs Xx}$.  Therefore, given
$w\in\Pcohp{\Limpl{\Locpcs Xx}{Y}}$, we have
$\Normsp{\Matappa w{(1-p)u}}Y\leq 1$ for all $u\in\Pcoh X$ and hence
$(1-p)w\in\Pcohp{\Limpl XY}$.

Let $t\in\Pcohp{\Limpl{\Excl X}{\One}}$. We have seen that, for all
$x\in\Pcoh X$ we have
$\Deriv t(x)\in\Pcohp{\Limpl{\Locpcs Xx}{\Locpcs\One{\Klfun
      t(x)}}}\subseteq \Pcohp{\Limpl{\Locpcs Xx}{\One}}$.  Therefore,
if we assume that $\Normsp xX\leq p$, we have
\begin{align}\label{eq:deriv-orth}
  (1-p)\Deriv t(x)\in\Pcohp{\Limpl X\One}=\Pcoh{\Orth X}\,.
\end{align}
Let $x\leq y\in\Pcoh X$ be such that $\Normsp yX\leq p$. Observe that
$2-p>1$ and that
\begin{align*}
  x+(2-p)(y-x)=y+(1-p)(y-x)\in\Pcoh X
\end{align*}
(because $\Norm y_X\leq p$ and $y-x\in\Pcoh X$).  We consider the
function
\begin{align*}
  h:\Intercc0{2-p}&\to\Intercc 01\\
  \theta&\mapsto\Klfun t(x+\theta(y-x))
\end{align*}
which is clearly analytic.
More precisely, one has $h(\theta)=\sum_{n=0}^\infty c_n\theta^n$ for
some sequence of non-negative real numbers $c_n$ such that
$\sum_{n=0}^\infty c_n(2-p)^n\leq 1$.

Therefore the derivative of $h$ is well defined on
$[0,1]\subset[0,2-p)$ and one has
\begin{align*}
  h'(\theta)=\Matappa{\Deriv
  t(x+\theta(y-x))}{(y-x)}\leq\frac{\Normsp{y-x}X}{1-p}
\end{align*}
by~\Eqref{eq:deriv-orth}, using Proposition~\ref{prop:chain-rule}. We have
\begin{align}\label{eq:lipsch-ordonne}
  0\leq\Klfun t(y)-\Klfun t(x)
  =h(1)-h(0)=\int_0^1h'(\theta)\,d\theta\leq\frac{\Normsp{y-x}X}{1-p}\,.
\end{align}

Let now $x,y\in\Pcoh X$ be such that $\Normsp xX,\Normsp yX\leq p$ (we
don't assume any more that $x$ and $y$ are comparable). We have
\begin{align*}
  \Absval{\Klfun t(x)-\Klfun t(y)}
  &=\Absval{\Klfun t(x)-\Klfun t(x\Infi y)+\Klfun t(x\Infi y)-\Klfun t(y)}\\
  &\leq \Absval{\Klfun t(x)-\Klfun t(x\Infi y)}
    +\Absval{\Klfun t(y)-\Klfun t(x\Infi y)}\\
  &\leq\frac 1{1-p}(\Normsp{x-(x\Infi y)}X+\Normsp{y-(x\Infi y)}X)\\
  &=\frac{\Distsp Xxy}{1-p}
\end{align*}
by~\Eqref{eq:lipsch-ordonne} since $x\Infi y\leq x,y$. So we have
proven the following result.
\begin{theorem}\label{th:pcoh-lipsch}
  Let $t\in\Pcohp{\Limpl{\Excl X}{\One}}$.  Given $p\in[0,1)$, the
  function $\Klfun t$ is Lipschitz with Lipschitz constant
  $\frac 1{1-p}$ on $\{x\in\Pcoh X\St\Normsp xX\leq p\}$ when $\Pcoh X$ is
  equipped with the distance $\Distspsymb X$, that is
  \begin{align*}
    \forall x,y\in\Pcoh X\quad \Normsp xX,\Normsp yX\leq p
    \Implies\Absval{\Klfun t(x)-\Klfun t(y)}\leq\frac{\Distsp Xxy}{1-p}\,.
  \end{align*}
\end{theorem}

\begin{remark}
  The Lipschitz constant cannot be uniformly upper-bounded on
  $\Pcoh X$, in particular it cannot be upper-bounded by $1$, that is
  $t$ is not always contractive. A typical example is
  $t=\phi_{0.5}\in\Pcoh(\Limpl{\Excl\One}\One)$ of
  Example~\ref{ex:expect-computation}: the left plot of
  Fig.~\ref{fig:example-expect} shows that the Lipschitz constant goes
  to $\infty$ when $p$ goes to $1$.
\end{remark}


\section{Application to the observational
  distance in $\PCFP$}
Given a $\PPCF$ term $M$ such that $\Tseq{}{M}{\Tnat}$, remember that
we use $\Probared M{\Num 0}$ for the probability of $M$ to reduce to
$\Num 0$ in the probabilistic reduction system
of~\cite{EhrhardPaganiTasson18}, so that
$\Probared M{\Num
  0}=\Proba{\Inistate{M}}{(\Eventconvz{\Inistate{M}})}$ with the
notations of Section~\ref{sec:PPCF-derivatives}. Remember that
$\Probared M{\Num 0}={\Psem M{}}_0$ by the Adequacy Theorem
of~\cite{EhrhardPaganiTasson18}.

Given a type $\sigma$ and two $\PCFP$ terms $M,M'$ such that
$\Tseq{}{M}{\sigma}$ and $\Tseq{}{M'}{\sigma}$, we define the
\emph{observational distance} $\Distobs{M}{M'}$ between $M$ and $M'$
as the sup of all the 
\begin{align*}
  \Absval{\Probared{\App CM}{\Num 0}-\Probared{\App C{M'}}{\Num 0}}  
\end{align*}
taken over terms $C$ such that $\Tseq{}{C}{\Timpl\sigma\Tnat}$ (testing
contexts).



If $\epsilon\in[0,1]\cap\Rational$ we have
$\Distobs{\Dice 0}{\Dice\epsilon}=1$ as soon as $\epsilon>0$.  It
suffices indeed to consider the context
\begin{align*}
  C=\Fix{\Abst f{\Timpl\Tnat\Tnat}{\Abst x\Tnat{\Ifv{x}{\APP fx}{z}{\Num
      0}}}}\,. 
\end{align*}
The semantics
$\Psem C{}\in\Pcohp{\Limpl{\Excl{\Snat}}{\Snat}}$ is a function
$c:\Pcoh{\Snat}\to\Pcoh{\Snat}$ such that
\begin{align*}
  \forall u\in\Pcoh{\Snat}\quad
  c(u)=u_0c(u)+\big(\sum_{i=1}^\infty u_i\big)\Snum 0
\end{align*}
and which is minimal (for the order relation of
$\Pcohp{\Limpl{\Excl{\Snat}}{\Snat}}$). If follows that
\begin{align*}
  c(u)=
  \begin{cases}
    0 & \text{if }u_0=1\\
    \frac 1{1-u_0}\sum_{i=1}^\infty u_i & \text{otherwise}\,.
  \end{cases}
\end{align*}
Then
\begin{align*}
  c((1-\epsilon)\Snum 0+\epsilon\Snum 1)=
  \begin{cases}
    0 & \text{if }\epsilon=0\\
    1 & \text{if }\epsilon>0\,.
  \end{cases}
\end{align*}
This is a well known phenomenon called ``probability amplification''
in stochastic programming.

Nevertheless, we can control a tamed version of the observational
distance. Given a closed $\PPCF$ term $C$ such that
$\Tseq{}{C}{\Timpl\sigma\Tnat}$ we define
\begin{align*}
  \Tamedc Cp=\Abst z\sigma{\App C{\If{\Dice p}{z}{\Loopt\sigma}}}
\end{align*}
and a tamed version of the observational distance is defined by
\begin{align*}
  \Tdistobs pM{M'}=\sup
  \left\{\Absval{\Probared{\App{\Tamedc Cp}M}{\Num 0}
  -\Probared{\App{\Tamedc Cp}{M'}}{\Num 0}}
  \St \Tseq{}{C}{\Timpl\sigma\Tnat}
\right\}\,.
\end{align*}
In other words, we modify our first definition of the observational
distance by restricting the universal quantification on contexts to
those which are of shape $\Tamedc Cp$.

\begin{theorem}\label{th:Tdistobs-denot-dist}
  Let $p\in[0,1)\cap\Rational$. Let $M$ and $M'$ be terms such that
  $\Tseq{}M{\sigma}$ and $\Tseq{}{M'}{\sigma}$.  Then we have
  \begin{align*}
    \Tdistobs p{M}{M'}\leq
    \frac{p}{1-p}\,\Distsp{\Tsem\sigma}{\Psem{M}{}}{\Psem{M'}{}}\,.
  \end{align*}
\end{theorem}
\begin{proof}
\begin{align*}
  \Tdistobs p{M}{M'}
  &=\sup\{\Absval{\Klfun{\Psem C{}}(p\Psem {M}{})_0
    -\Klfun{\Psem C{}}(p\Psem {M'}{})_0}
    \St\Tseq{}{C}{\Timpl\sigma\Tnat}\}\\
  &\leq\sup\{\Absval{\Klfun{t}(p\Psem{M}{})
    -\Klfun{t}(p\Psem {M'}{}})
    \St t\in\Pcohp{\Limpl{\Excl{\Tsem{\sigma}}}\One}\}\\
  &\leq
    \frac{\Distsp{\Tsem\sigma}{p\Psem{M}{}}{p\Psem{M'}{}}}{1-p}
    =\frac{p}{1-p}\,\Distsp{\Tsem\sigma}{\Psem{M}{}}{\Psem{M'}{}}\,.
\end{align*}
by the Adequacy Theorem and by Theorem~\ref{th:pcoh-lipsch}.
\end{proof}


Since $p/(1-p)=p+p^2+\cdots$ and $\Distsp{\Tsem\sigma}{\_}{\_}$ is an
over-approximation of the observational distance restricted to linear
contexts, this inequation carries a rather clear operational intuition
in terms of execution in a Krivine machine as in
Section~\ref{sec:PCF-operational} (thanks to Paul-André Melliès for
this observation).
%
%
Indeed, using the stacks of Section~\ref{sec:PCF-operational}, a
\emph{linear} observational distance on $\PPCF$ terms can easily be
defined as follows, given terms $M$ and $M'$ such that
$\Tseq{}{M}{\sigma}$ and $\Tseq{}{M'}{\sigma}$:
\begin{align*}
  \Ldistobs{M}{M'}=
  \sup_{\Tseqst\sigma\pi}\Absval{\Proba{\State M\pi}
  {(\Eventconvz{\State M\pi})}-
  \Proba{\State{M'}\pi}{(\Eventconvz{\State{M'}\pi})}}\,.
\end{align*}
In view of Theorem~\ref{th:Tdistobs-denot-dist} and of the fact that
$\Ldistobs{M}{M'}\leq\Distsp{\Tsem\sigma}{\Psem{M}{}}{\Psem{M'}{}}$
(easy to prove, since each stack can be interpreted as a linear
morphism in $\PCOH$), a natural and purely syntactic conjecture is that
\begin{align}\label{eq:dist-ineq-strong}
  \Tdistobs p{M}{M'}\leq\frac{p}{1-p}\,\Ldistobs{M}{M'}\,.
\end{align}
This seems easy to prove in the case
$\Proba{\State{M'}\pi}{(\Eventconvz{\State{M'}\pi})}=0$: it suffices
to observe that a path which is a successful reduction of
$\State{\App{\Tamedc Cp}M}{\Stempty}$ in the ``Krivine Machine'' of
Section~\ref{sec:PCF-operational} (considered here as a Markov chain)
can be decomposed as
\begin{align*}
  \State{\App{\Tamedc Cp}M}{\Stempty}
  &\Rel{\to^*} \State{\If{\Dice p}{M}{\Loopt\sigma}}{\pi_1(C,M)}
  \Rel{\to^*} \State{\If{\Dice p}{M}{\Loopt\sigma}}{\pi_2(C,M)}\\
  &\Rel{\to^*}\cdots\Rel{\to^*}
    \State{\If{\Dice p}{M}{\Loopt\sigma}}{\pi_k(C,M)}
    \Rel{\to^*}\State{\Num 0}{\Stempty}
\end{align*}
where $(\pi_i(C,M))_{i=1}^k$ is a finite sequence of stacks such that
$\Tseqst\sigma{\pi_i(M)}$ for each $i$. Notice that this sequence of
stacks depends not only on $C$ and $M$ but also on the considered path
of the Markov chain.

In the general case, Inequation~\Eqref{eq:dist-ineq-strong} seems less
easy to prove because, for a given common initial context $C$, the
sequences of reductions (and of associated stacks) starting with
$\State{\App{\Tamedc Cp}M}{\Stempty}$ and
$\State{\App{\Tamedc Cp}{M'}}{\Stempty}$ differ. This divergence has
low probability when $\Ldistobs{M}{M'}$ is small, but it is not
completely clear how to evaluate it. Coinductive methods like
probabilistic bisimulation as in the work of Crubillé and Dal~Lago are
certainly relevant here.

Our Theorem~\ref{th:pcoh-lipsch} shows that another and more
geometric approach, based on a simple denotational model, is also
possible to get Theorem~\ref{th:Tdistobs-denot-dist} which, though
weaker than Inequation~\Eqref{eq:dist-ineq-strong}, allows
nevertheless to control the $p$-tamed distance.

We finish the paper by observing that the equivalence relations
induced on terms by these observational distances coincide with the
ordinary observational distance if $p\not=0$.

\begin{theorem}\label{th:distobsp-controle}
  Assume that $0<p\leq1$. If $\Tdistobs p M{M'}=0$ then $M\Obseq M'$
  (that is, $M$ and $M'$ are observationally equivalent).
\end{theorem}
\begin{proof} If $\Tseq{}M\sigma$ we set
$M_p=\If{\Dice p}{M}{\Loopt\sigma}$.  If $\Tdistobs p M{M'}=0$ then
$M_p\Obseq M'_p$ by definition of observational equivalence, hence
$\Psem{M_p}{}=\Psem{M'_p}{}$ by our Full Abstraction
Theorem~\cite{EhrhardPaganiTasson18}, but $\Psem{M_p}{}=p\Psem M{}$
and similarly for $M'$. Since $p\not=0$ we get $\Psem M{}=\Psem{M'}{}$
and hence $M\Obseq M'$ by adequacy~\cite{EhrhardPaganiTasson18}.
\end{proof}

So for each $p\in(0,1)$ and for each type $\sigma$ we can consider
$\Tdistobsf p$ as a distance on the observational classes of closed
terms of type $\sigma$. We call it the \emph{$p$-tamed observational
  distance}. Our Theorem~\ref{th:Tdistobs-denot-dist} shows that we
can control this distance using the denotational distance. For
instance we have
\(
  \Tdistobs p{\Dice 0}{\Dice\epsilon}\leq\frac{2p\epsilon}{1-p}
  \)
so that $\Tdistobs p{\Dice 0}{\Dice\epsilon}$ tends to $0$ when
$\epsilon$ tends to $0$.

\section{Conclusion}
The two results of this paper are related: both use derivatives
wrt.~probabilities to evaluate the number of times arguments are
used. The derivatives used in the second part are more general than
those of the first part simply because the parameters wrt.~which
derivatives are taken can be of an arbitrary type whereas in the first
part, they are of ground type, but this is essentially the only
difference.

Indeed in Section~\ref{sec:PPCF-derivatives} we computed partial
derivatives of morphisms
\begin{align*}
  t\in\PCOH(\Excl X,\Snat)
\end{align*}
where $X=\Snat\IWith\cdots\IWith\Snat$ ($k$ copies). More precisely,
the $t$'s we consider in that section are such that if
$t_{\Vect\mu,n}\not=0$ then $n=0$ and each $\mu_i\in\Mfin\Nat$
satisfies $\Supp{\mu_i}\subseteq\Eset 0$. So actually we can consider
such a $t$ as a an element of
$\PCOH(\Exclp{\One\IWith\cdots\IWith\One},\One)$ which induces a
function
\begin{align*}
  \Klfun t:\Pcohp{\One\IWith\cdots\IWith\One}\Isom\Intercc01^k
  &\to\Pcoh\One\Isom\Intercc 01\\
  x &\mapsto\sum_{\List n1k\in\Nat}t_{\List n1k}\prod_{i=1}^k x_i^{n_i}
\end{align*}
where $t_{\List n1k}\in\Realp$ for each $\List n1k\in\Nat$.  Let
$Y=\One\IWith\cdots\IWith\One$. Let $\Vect x\in\Pcoh\Snat^k$ and
$x\in\Pcoh Y$ so that $x$ can be seen as a tuple
$(x_1,\dots,x_k)\in\Intercc 01$ and $\Norm x_Y=\max_{i=1}^kx_i$. It
follows that $\Locpcs Yx$ can be described as follows:
\begin{align*}
  \Web{\Locpcs Yx}&=\Eset{i\St 1\leq i\leq k\text{ and }x_i<1}\\
  \Pcohp{\Locpcs Yx}&=\Eset{u\in\Realpto{\Web{\Locpcs Yx}}\St
                     \forall i\in\Web{\Locpcs Yx}\ x_i+u_i\leq 1}\,.
\end{align*}
and then we have defined the differential of $t$,
\begin{align*}
  \Deriv t(x)\in\PCOH(\Locpcs Yx,\One)
\end{align*}
in Section~\ref{sec:loca-PCS-diff}. This differential relates as
follows with the partial derivatives used in
Section~\ref{sec:PPCF-derivatives}:
\begin{align*}
  \forall u\in\Pcoh{\Locpcs Yx}\quad
  \Matappa{\Deriv t (x)}u=\sum_{i\in\Web{\Locpcs Yx}}t'_i(x)u_i
\end{align*}
where $t'_i(x)\in\Realp$ is the $i$th partial derivative of the
function $\Klfun t$ at $x$. Notice that in
Section~\ref{sec:PPCF-derivatives} we use slightly more general partial
derivatives computed also at indices $i$ such that $x_i=1$ (which are
actually left derivatives since the function can be undefined for
$x_i>1$) where they can take infinite values, in accordance with the
fact that the expectation of the number of steps can be $\infty$ like
in the example $\phi_{0.5}$, see Fig.~\ref{fig:example-expect}. This
is not allowed in Section~\ref{sec:loca-PCS-diff} where we insist on
keeping all derivatives finite for upper-bounding them.

We think that these preliminary results provide motivations for
investigating further differential extensions of $\PPCF$ and related
languages in the spirit of the differential
lambda-calculus~\cite{EhrhardRegnier02}.

\section*{Acknowledgments}
We thank Raphaëlle Crubillé, Paul-André Melliès, Michele Pagani and
Christine Tasson for many enlightening discussions on this work. We
also thank the referees of the FSCD'19 version of this paper for their
precious comments and suggestions. Last but not least we thank warmly
the reviewers of this journal version for their in-depth reading and
understanding of the paper and for their invaluable help in improving
the presentation.

This research was partly funded by the ANR project ANR-19-CE48-0014
\emph{Probabilistic Programming Semantics (PPS)}.



\bibliographystyle{alpha}
\bibliography{newbiblio}




\end{document}